\documentclass[11pt]{article}
\usepackage{color}
\usepackage[colorlinks=true, allcolors=blue,backref=page]{hyperref}
\usepackage{amsmath, amssymb, amsthm, thmtools}
\usepackage{mathrsfs}
\usepackage{mathtools}
\usepackage[noabbrev,capitalize,nameinlink]{cleveref}
\usepackage{fullpage}
\usepackage[noadjust]{cite}
\usepackage{graphics}
\usepackage{pifont}
\usepackage{tikz}
\usetikzlibrary{decorations.markings,shapes.geometric}
\usepackage{bbm}
\usepackage{float} 
\usepackage[T1]{fontenc}

\usetikzlibrary{arrows.meta}

\usepackage{environ}
\usepackage{framed}
\usepackage{url}
\usepackage[linesnumbered,ruled,vlined]{algorithm2e}
\usepackage[noend]{algpseudocode}
\usepackage[labelfont=bf]{caption}
\usepackage{cite}
\usepackage{framed}
\usepackage[framemethod=tikz]{mdframed}
\usepackage{appendix}
\usepackage{graphicx}
\usepackage{tcolorbox}
\usepackage{enumerate}
\allowdisplaybreaks[1]
\usepackage{enumerate}
\usepackage{stmaryrd}

\usepackage[margin=1in]{geometry}
\usepackage{braket}

\usepackage[shortlabels]{enumitem}
\crefformat{enumi}{#2#1#3}
\crefrangeformat{enumi}{#3#1#4 to~#5#2#6}
\crefmultiformat{enumi}{#2#1#3}
{ and~#2#1#3}{, #2#1#3}{ and~#2#1#3}

\DeclareSymbolFont{symbolsC}{U}{pxsyc}{m}{n}
\SetSymbolFont{symbolsC}{bold}{U}{pxsyc}{bx}{n}
\DeclareFontSubstitution{U}{pxsyc}{m}{n}
\DeclareMathSymbol{\medcircle}{\mathbin}{symbolsC}{7}

\crefformat{equation}{#2(#1)#3}
\crefrangeformat{equation}{#3(#1)#4 to~#5(#2)#6}
\crefmultiformat{equation}{#2(#1)#3}{ and~#2(#1)#3}{, #2(#1)#3}{ and~#2(#1)#3}
\AtBeginEnvironment{appendices}{\crefalias{section}{appendix}} %

\numberwithin{equation}{section}
\newtheorem{theorem}{Theorem}[section]
\newtheorem{proposition}[theorem]{Proposition}
\newtheorem{lemma}[theorem]{Lemma}

\newtheorem{corollary}[theorem]{Corollary}

\newtheorem*{question*}{Question}

\theoremstyle{definition}
\newtheorem{definition}[theorem]{Definition}

\newtheorem*{definition*}{Definition}
\newtheorem*{theorem*}{Theorem}

\newtheorem{remark}[theorem]{Remark}

\newcommand{\mb}[1]{\mathbb{#1}}
\newcommand{\mc}[1]{\mathcal{#1}}
\newcommand{\mbf}[1]{\mathbf{#1}}

\newcommand{\on}{\operatorname}

\newcommand{\eps}{\varepsilon}

\let\R\Real

\let\F\Field
\let\FF\Field

\newcommand{\smcdef}{\sum_{j=0}^{\ell}c_j \alpha^j}
\newcommand{\smudef}{\sum_{j=0}^{\ell}u_j \alpha^j}

\let\originalleft\left
\let\originalright\right
\renewcommand{\left}{\mathopen{}\mathclose\bgroup\originalleft}
\renewcommand{\right}{\aftergroup\egroup\originalright}

\allowdisplaybreaks

\title{Locality of Curve-Decoding and Improved Proximity Gaps}
\begin{document}
\author{
Rohan Goyal\thanks{Massachusetts Institute of Technology. \texttt{rohan\_g@mit.edu}}\and 
Venkatesan Guruswami \thanks{University of California, Berkeley. \texttt{venkatg@berkeley.edu}}\and 
Yihang Sun\thanks{Stanford University.  \texttt{kimisun@stanford.edu}}\and  Mary Wootters\thanks{Stanford University. \texttt{marykw@stanford.edu}}}
\date{}
\maketitle

\begin{abstract}

Proximity gaps are a property of error correcting codes that arise in the study of Interactive Oracle Proofs (IOPs) and Succinct Non-interactive Arguments of Zero Knowledge (SNARKs).  Informally, we say that a code $C \subset  \Sigma^n$ exhibits a \emph{proximity gap} (with respect to degree-$\ell$ curves) if for any degree-$\ell$ curve $u(x) \in \Sigma^n$, either \emph{every} point on $u(x)$ is close to $C$, or else \emph{most} of them are far from $C$.

Recent work~\cite{gg25a} has established near-optimal proximity gaps for many families of codes, including subspace design codes, as well as random ensembles like random linear codes, Reed-Solomon codes with random evaluation points, and Gallager's ensemble of LDPC codes.  However, the parameters for these latter randomized ensembles are worse than the parameters for subspace design codes, and degrade as the degree $\ell$ increases.

In this work, we obtain improved proximity gaps for random ensembles of codes, including random linear codes, Reed-Solomon codes with random evaluation points, and Gallager's ensemble.  Quantitatively, our results for these random ensembles match the results that \cite{gg25a}  attained for subspace design codes.  In fact, our techniques are a black-box transference from subspace design codes: Any progress on subspace design codes will automatically lead to analogous progress for these random ensembles.

To obtain our results, we extend the Local Coordinate-wise Linear (LCL) property framework developed in \cite{lms25,bcdz25b} to a \emph{row-span constrained} version.  This allows us to cast \emph{curve-decodability}---a property that implies proximity gaps---directly as an (row-span constrained) LCL property, and make use of that machinery. In contrast, because curve-decodability is not obviously a (vanilla) LCL property, prior work had worked with a proxy property instead, leading to the aforementioned parameter losses. {In addition, we extend the framework to also show an equivalence theorem for Gallager's ensemble of random LDPC codes and random linear codes for our row-span constrained LCL properties.}
\end{abstract}

\newpage
\tableofcontents
\newpage
\section{Introduction}
\label{sec:intro}

Proximity gaps are a property of error correcting codes that arise in the study of Interactive Oracle Proofs (IOPs) and Succinct Non-interactive Arguments of Zero Knowledge (SNARKs), both of which are relevant for blockchains and related technologies. Let $\mathbb{F}$ be a finite field, and let $\Sigma$ be a vector space over $\F$.  An \emph{$\F$-additive code} $C$ (with alphabet $\Sigma$ and block length $n$) is an $\F$-linear subspace $C \subset  \Sigma^n$.  Informally, we say that $C$ exhibits a \emph{proximity gap} (for degree-$\ell$ curves, which is the focus of this paper) if the following holds.  For any degree-$\ell$ curve $u(x) = \sum_{i=0}^\ell u_i x^i$ in $\Sigma^n$ given by $u_0, \ldots, u_\ell \in \Sigma^n$ {and $x\in\mb{F}$}, either \emph{every} point $u(\alpha)$ on the curve is close to $C$, or else a $1 - \eps$ fraction of the points on the curve are far from $C$.  Formally, we have the following definition.  Below, $\Delta(x,y) := \frac{1}{n} \sum_{i=1}^n \mathbf{1}[x_i \neq y_i]$ is the (relative) Hamming distance between $x,y \in \Sigma^n$, {and $\Delta(x, C):=\min_{y\in C}\Delta(x, y)$ for code $C$.}

\begin{definition}[Proximity Gaps (PG)]\label{def: curve PG}
   An $\FF$-additive code $C\subset  \Sigma^n$ is said to have $(\ell, \delta, \eps, \gamma)$ proximity gap if for all $u_0, u_1, \cdots, u_\ell \in \Sigma^n$, and for all $\delta'\le \delta$, \[\mb{P}_{\alpha\in \FF}\left(\Delta\left(\sum_{j=0}^\ell u_j \alpha^j, C\right)\le \delta'\right)>\eps \implies \forall \alpha \in \FF, \ \Delta\left(\sum_{j=0}^\ell u_j \alpha^j, C\right)\le \frac{\delta'}{1-\gamma}\ .   \]
   When $\gamma = 0$, we say that $C$ has an $(\ell, \delta, \eps)$ proximity gap.  
\end{definition}

Proximity gaps, especially for \emph{Reed-Solomon (RS) Codes},\footnote{Reed-Solomon codes are a classical family of codes based on low-degree polynomials.  Given fixed distinct evaluation points $\alpha_1, \ldots, \alpha_n \in \mathbb{F}_q =: \Sigma$, the corresponding Reed-Solomon code $C \subset  \Sigma^n$  of length $n$ and dimension $k \leq n$  is given by $C = \{ (f(\alpha_1), \ldots, f(\alpha_n)) : f \in \mathbb{F}_q[x], \mathrm{deg}(f) < k\}$. } 
have been a hot topic lately, with the Ethereum Foundation offering \$1M in prizes for resolving (or disproving) conjectures related to proximity gaps for RS codes~\cite{proximityprize}.  See \cite{ABF26} for a recent survey on proximity gaps and problems relevant to the proximity prize.

Until recently, the best proximity gap guarantees required $\delta < 1 - \sqrt{R}$, where $R = \frac{\log_{|\Sigma|}|C|}{n}$ is the \emph{rate} of the code $C \subset  \Sigma^n$.  In recent work \cite{gg25a}, Guruswami and Goyal established near-optimal proximity gaps---with $\delta$ approaching $1 - R$---for several families of codes.  For \emph{Folded Reed-Solomon Codes, Multiplicity Codes},\footnote{Informally, a \emph{folded Reed-Solomon code} is obtained by ``folding'' (aka, bundling symbols together) an RS code $C \in \F_q^n$ to obtain a code $C' \in (\F_q^m)^{n/m}$.  A (univariate) \emph{multiplicity code} is similar to an RS code, except each symbol contains not just $f(\alpha_i)$, but also $f^{(j)}(\alpha_i)$ for $j = 1, 2, \ldots, m-1$, where $f^{(j)}$ denotes the $j$'th (Hasse) derivative.  In either case, the alphabet size increases from $q$ to $q^m$.
By a result of \cite{GK16}, both are special cases of Subspace Design Codes, defined in \cref{def: Subspace-designs}.} and more generally any \emph{Subspace Design Code} (\cref{def: Subspace-designs}), they established $(\ell, \delta, \eps)$ proximity gaps with $\delta = 1 - R - 2\eta$ for any $\eta > 0$, provided that
\begin{equation}\label{eq:target}
\eps q \geq \frac{n \ell}{\eta} + O\left( \frac{\ell^2}{\eta^3} \right),
\end{equation}
where $q$ is the size of the base field.  Thus, the left-hand side $\eps q$ is a bound on the number of points on the curve $u(x)$ that can still be close to $C$, in the case that not all $q$ points are.  We'd like for this bound to be as small as possible.

Via a framework from \cite{bcdz25b,lms25} that establishes a connection between random linear codes and subspace design codes, the work \cite{gg25a} extended their results for subspace design codes to several random ensembles of codes, including random linear codes, Reed-Solomon codes with random evaluation points, and Gallager's ensemble of random LDPC codes.  However, some parameters degraded in this extension: The final result that \cite{gg25a} established for these random ensembles was $(\ell, \delta, \eps)$ proximity gaps with $\delta = 1 - R - 2\eta$, provided that
\begin{equation}\label{eq:previous}
\eps q \geq \frac{n \ell(1-R)}{\eta} + \left( \frac{\ell}{\eta^2}\right)^{O(\ell)}.
\end{equation}
When $\ell$ is large, on the order of $\Omega(\log n)$, the error term $(\ell/\eta^2)^{O(\ell)}$ in \cref{eq:previous} swamps the main term, and the requirement on $\eps q$ quickly becomes large. 
This leaves us with the question of whether a bound like \cref{eq:target} is achievable for these random ensembles of codes.  This question has practical as well as theoretical motivations: Larger values of $\ell$ can occur in practice, for example in WHIR~\cite{whir}; random ensembles of codes are also practically interesting, for example in Blaze~\cite{bcfrrz25} or Bolt~\cite{gnr26};\footnote{In more detail, Blaze considers a random ensemble of \emph{RAA} codes, which use random permutations to decide the parity checks.  Bolt uses \emph{Gallager's ensemble} of LDPC codes, which we also study in this paper.} and understanding RS codes with random evaluation points may deepen our understanding of explicit RS codes, relevant to the proximity prize~\cite{proximityprize}.

\subsection{Our Contributions}
\label{sec:main}
Our main contribution is to improve the bound \cref{eq:previous} for random ensembles of codes to match the bound \cref{eq:target} that was already established for subspace design codes.  We establish these improved bounds for random linear codes, RS codes with random evaluation points, and  {Gallager's ensemble} of LDPC codes.  See \cref{def: RLC,def: RRS,def: RLDPC} for formal definitions of these code families.

\begin{theorem}
[Informal; see \cref{thm:rlc-main,thm:rrs-main,thm:rldpc-main}]\label{thm:main}
Fix any $\ell\in\mb{N}$ and $\eta>0$. Let $C\in \mb{F}_q^n$ be a random linear code, random Reed-Solomon code, or random LDPC code\footnote{Technically, the result for random LDPC codes requires the leading term $n\ell (1-R)/\eta$ to be slightly larger; it is stated in \cref{thm:rldpc-main} as $2n \ell (1-R)/\eta$; the factor of two can be improved to any constant larger than one (see \cref{rem:factor of two}).} of rate $R\in (0, 1)$.
Then, with probability at least $2/3$ over the choice of $C$, $C$ has $(\ell, 1-R-2\eta, \eps)$ proximity gaps (with respect to degree-$\ell$ curves), provided that
\begin{equation}\label{eq:good-error-2}
\eps q \ge \frac{n\ell(1-R)}{\eta}+O\left(\frac{\ell^2}{\eta^3}\right).
\end{equation}   
\end{theorem}
In fact, our results establish not just proximity gaps, but also stronger notions known as \emph{Correlated Agreement} (CA, \cref{def: curve CA}) and \emph{Mutual Correlated Agreement} (MCA, \cref{def: curve MCA}); see \cref{thm:rlc-main,thm:rrs-main,thm:rldpc-main} for the full statements.  

Our framework, discussed more below in \cref{sec:overview}, is a black-box transference method from subspace design codes.  Thus, any future improvements for subspace design codes will immediately lead to an analogous improvement in \cref{thm:main}.

In order to establish our results, we generalize the \emph{Local Coordinate-wise Linear} (LCL) framework of \cite{lms25} to a new ``row-constrained'' version of the framework, described below.  We hope that this generalization may find other applications going forward.

\begin{remark}
    {\cite{lms25} showed an equivalence for LCL properties between random linear and random Reed-Solomon codes, which we generalize to the row-constrained setting. To extend our results to random LDPC codes, we prove an analogous equivalence theorem of our generalized LCL properties for Gallager's ensemble of random LDPC codes and random linear codes based on \cite{mrrsw20}. See \cref{subsec: random LDPC} and \cref{app: RLDPC} for more details.}
\end{remark}

\subsection{Techniques: Curve-Decoding as a Row-Span Constrained LCL Property}
\label{sec:overview}

Following \cite{gg25a}, we study \emph{curve-decoding}.  For the rest of this discussion, assume for simplicity that the alphabet $\Sigma$ is equal to $\mathbb{F}_q$, so $C$ is a linear subspace of $\mathbb{F}_q^n$.  

Informally, such a code $C$ is $(\ell, \delta, a, b)$ \emph{curve-decodable} if, for any degree-$\ell$ curve $u(x) = \sum_{i=0}^\ell u_i x^i$ with $u_i \in \Sigma^n$, the following holds.  Suppose that there is a set $A \subset  \mathbb{F}_q$ of size $a$ so that $u(\alpha)$ is $\delta$-close to some codeword $f(\alpha) \in C$ for all $\alpha \in A$.  Then, there should be a ``codeword curve'' $c(x) = \sum_{i=0}^\ell c_i x^i$, where $c_i \in C \subset  \Sigma^n$, so that $f(\alpha) = c(\alpha)$ for at least $b$ values of $\alpha \in A$.  See \cref{def: curve-decodability} for a formal definition; the property is illustrated in \cref{fig:curve-decoding}. 

The work \cite{gg25a} showed that curve-decodability is a sufficient condition to establish proximity gaps (as well as stronger notions like correlated agreement and mutual correlated agreement); see \cref{thm: CA from only curves} and \cref{thm: appendix MCA from curves}.  They then established that subspace design codes have curve-decodability (\cref{thm:sdc-good-cd}), which led to their results about subspace design codes.

Our work comes in the next step.  To obtain their results for random ensembles of codes, \cite{gg25a} applied a framework of \cite{bcdz25b,lms25}.  This framework roughly shows that any \emph{Local Coordinate-wise Linear (LCL)} property that is satisfied by subspace design codes is also satisfied by these random ensembles.  We will explain what LCL properties are in a moment, but the main challenge with this approach is that ``being curve-decodable'' is in fact \emph{not} (or, not obviously) an LCL property.  To get around this, \cite{gg25a} defined another property, which they called \emph{$\bigvee$-decoding}, which \emph{is} captured by the LCL framework. Then existing results imply that the random ensembles of codes are just as $\bigvee$-decodable as subspace design codes, and translating back to curve-decoding gives the final result. Unfortunately, going through $\bigvee$-decoding results in a loss in parameters, degrading \eqref{eq:target} to \eqref{eq:previous}.  We avoid this loss by \emph{extending} the LCL framework to apply directly to curve-decoding.

\begin{figure}
    \centering
    \begin{tikzpicture}[scale=.8]
        \begin{scope}

  \draw[blue, thick, smooth, variable=\t, domain=-1.5:1.5, samples=200]
    plot ({2.2*\t}, {\t*\t*\t - 1.2*\t});
    
     \foreach \t in {-1.3,-1.0,0,0.5} {
    \node[orange, star, star points=5, fill=orange, draw=orange, inner sep=1.6pt] at ({1.7*\t}, {1.3*\t*\t*\t - 1.6*\t + .4}) {};
  }

  \node[star, star points=5, fill=orange, draw=orange, inner sep=1.6pt] at (2.7,0) {};

  \foreach \t in {-1.4,-1.2,-0.7,-.2,0.3,.7,1.1} {
    \fill[blue] ({2.2*\t}, {\t*\t*\t - 1.2*\t}) circle (2.8pt);
  }
  \node[blue] at (-4,-1.5) {$u(x)$};
  \node[blue] at (-3.3, -.2) {\footnotesize $u(\alpha)$};
  \node[orange, anchor=west] at (-2.3,-.6) {\footnotesize $f(\alpha) \in C$};

        \end{scope}
\node at (4,0) {\huge $\Rightarrow$};
    \begin{scope}[xshift=9cm]
  \node[star, star points=5, fill=orange, draw=orange, inner sep=1.6pt] at (2.7,0) {};
  \draw[blue, thick, smooth, variable=\t, domain=-1.5:1.5, samples=200]
    plot ({2.2*\t}, {\t*\t*\t - 1.2*\t});

\node[blue] at (-4,-1.5) {$u(x)$};
\node[orange] at (3.1,2.3) {$c(x)$};
  \node[blue] at (-3.3, -.2) {\footnotesize $u(\alpha)$};
  \node[orange, anchor=west] at (-2.3,-.6) {\footnotesize $f(\alpha) =c(\alpha)$};

  \draw[orange, thick, dashed, smooth, variable=\t, domain=-1.5:1.5, samples=200]
    plot ({1.7*\t}, {1.3*\t*\t*\t - 1.6*\t + .4});
    
     \foreach \t in {-1.3,-1.0,0,0.5} {
    \node[orange, star, star points=5, fill=orange, draw=orange, inner sep=1.6pt] at ({1.7*\t}, {1.3*\t*\t*\t - 1.6*\t + .4}) {};
  }

       \foreach \t in {-1.45,-0.5,1.0,1.3} {
    \node[orange, star, star points=5, fill=white, draw=orange, inner sep=1.6pt] at ({1.7*\t}, {1.3*\t*\t*\t - 1.6*\t + .4}) {};
  }

  \foreach \t in {-1.4,-1.2,-0.7,-.2,0.3,.7,1.1} {
    \fill[blue] ({2.2*\t}, {\t*\t*\t - 1.2*\t}) circle (2.8pt);
  }
  \node[blue] at (-4,-1.5) {$u(x)$};

        \end{scope}
    \end{tikzpicture}

    \caption{Depiction of curve decodability (\cref{def: curve-decodability}).  The orange stars represent $f(\alpha) \in C$ for $\alpha \in A$, and the blue dots represent $u(\alpha) \in \Sigma^n$ for $\alpha \in \mathbb{F}_q$, where $u(x)$ is a degree-$\ell$ curve.  Informally, a code $C \subseteq \Sigma^n$ is $(\ell, \delta, a,b)$-curve-decodable if the following holds: If there are at least $a$ orange stars $\delta$-close to distinct blue dots, then at least $b$ of those orange stars lie on a degree-$\ell$ curve $c(x)$.  In the picture, $a = 5$, and $b = 4$.}\label{fig:curve-decoding}
\end{figure}
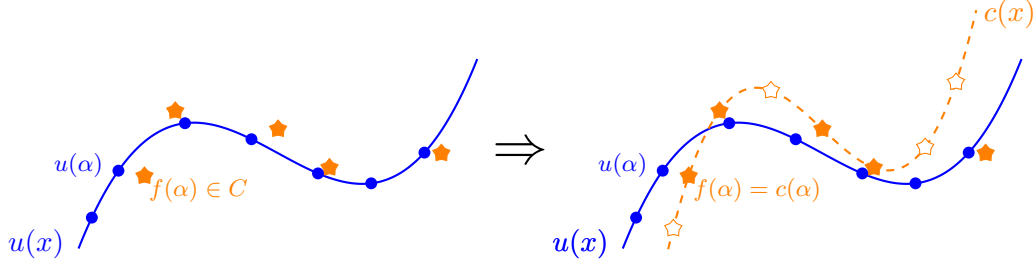

\medskip\noindent\textbf{The LCL Framework.}
Before explaining our extension, we first explain the LCL framework of \cite{lms25}.  Intuitively, an $r$-local LCL property is one defined by the exclusion of $r$-sized bad sets of codewords.  For example, the property of ``having good distance''  might be defined by the exclusion of all pairs ($r=2$) of codewords that agree in too many places.  The LCL framework allows for any definition of ``bad'' that can be captured by a family $\mathcal{F}$ of coordinate-wise linear relationships.  

In more detail, let $C \subset  \mathbb{F}_q^n$ be a linear code, and consider a set of $r$ distinct ``bad'' codewords, stacked together to form the columns of a matrix $M \in \F_q^{n \times r}$.
Let $\mathcal{F}$ be a family of \emph{profiles} $\mathbf{V} = (V_1, \ldots, V_n)$, where each $V_i \in \mathcal{L}(\mathbb{F}_q^r)$.  (Here, $\mathcal{L}(W)$ denotes the collection of subspaces of $W$.)
In the LCL framework, the ``badness'' of a matrix $M$ is captured by the following criterion: There is some $\mathbf{V} = (V_1, \ldots, V_n) \in \mathcal{F}$ so that for all $i \in [n]$, the $i$'th row of $M$ is contained in $V_i$. 

Thus, an LCL property defined by $\mathcal{F}$ is the property of \emph{avoiding} all ``bad'' matrices $M$ that satisfy the constraints given by some $\mathbf{V} \in \mathcal{F}$ (we say that $C$ \emph{avoids} $M$ if the column span of $M$ is not contained in $C$).
See \cref{def: containing a local profile} for a formal definition.

\medskip\noindent\textbf{Capturing Curve-Decoding and the Approach of \cite{gg25a}.}
The LCL framework captures many familiar properties, for example distance, list-decoding, and list-recovery.  Unfortunately, it does not seem to capture curve-decoding.  To see why, let us consider two approaches.

For the first---and probably most natural---approach, suppose that we have a bad example for $(\ell, \delta, a, b)$ curve-decoding.  That is, suppose we have a set $A \subset  \mathbb{F}_q$ of size $a$, a degree-$\ell$ curve $u(x) \in \mathbb{F}_q^n$, and a function $f: \mathbb{F}_q \to C$ so that $\Delta(u(\alpha), f(\alpha)) \leq \delta$ for all $\alpha \in A$; but suppose that there is \emph{no} degree-$\ell$ ``codeword curve'' $c(x) = \sum_{i=0}^\ell c_i x^i$ for $c_i \in C$ that passes through at least $b$ of the $f(\alpha)$.

We will attempt to construct a ``witness'' matrix $M \in \mathbb{F}_q^{n \times A}$ whose columns are $f(\alpha)$ for all $\alpha \in A$.  In order
to certify this as a bad example, we should be able to come up with some LCL profile $\mathbf{V} = (V_1, \ldots, V_n)$ that ``explains'' the badness of the example.  There are two properties that make a ``bad example'' bad:
\begin{itemize}
    \item[(A)] Each of the $f(\alpha)$ (the columns of $M$) agree in many coordinates with some ambient curve $u(x)$ of degree $\ell$.
    \item[(B)] No $b$ columns of $M$ themselves lie on the same degree-$\ell$ curve.
\end{itemize}
Item (A) can be captured with an $a$-local profile $\mathbf{V} = (V_1, \ldots, V_n)$.  Indeed, let $A_i = \{ \alpha \in A \,:\, u(\alpha)_i = f(\alpha)_i \}$, and then define $V_i$ to be the subspace of all $y \in \mathbb{F}_q^A$ so that $y$ (when viewed as a function $y:A \to \mathbb{F}_q$) agrees with some degree-$\ell$ polynomial on $A_i$.  Then consider the class $\mathcal{F}$ of all profiles $\mathbf{V}$ generated this way (over all appropriate choices of $A_i$).  Item (A) holds if and only if $C$ contains some $\mathbf{V} \in \mathcal{F}$.

However, it is not clear how to capture item (B) with the LCL framework, as (B) is a more global property of the matrix, not easily captured with coordinate-wise constraints.  Thus, another approach seems needed.

As a second possible approach, we describe what \cite{gg25a} did, and briefly explain why it incurs an exponential loss in the parameter $\ell$.
The idea is to consider first the special case that $b = \ell + 2$.   Given a bad example for curve-decoding, consider a much larger ``witness'' matrix $M \in \mathbb{F}_q^{n \times r}$ for $r = {a \choose \ell +1}$, as follows.  The columns of $M$ now correspond to subsets $S \in {A \choose \ell+1}$, and the columns correspond to polynomials $g_S(x)$ defined by interpolating through the points $(\alpha, f(\alpha))$ for all $\alpha \in S$.  It turns out that if $b = \ell+2$, then all of these polynomials must be distinct: Otherwise, if $g_S = g_{S'}$ for some $S \neq S'$, then this gives a degree-$\ell$ curve passing through at least $|S \cup S'| \geq \ell+2 = b$ of the codewords $f(\alpha)$, contradicting the badness of this example.  Then, if $\mathbb{F}_q$ is large enough, there is some $\beta$ so that $\left\{g_S(\beta) \in C\,:\, S \in {A \choose \ell+1}\right\}$ are also all distinct, and these can be the columns of $M$.  However, this requires us to take $\ell = b+2$; otherwise, we can't be guaranteed that the columns are distinct, which is required by the LCL framework.  The insight of \cite{gg25a} was to essentially transform the general-$\ell$ case to the $\ell = b+2$ case via a new property called $\bigvee$-decoding; but this came at a cost of an exponential dependence on $\ell$.  (We emphasize that the exponential loss comes from the transformation to the $b = \ell+2$ case, not the larger size of the witness matrix, although we will be able to overcome both issues in this work.)

\medskip\noindent\textbf{Our Approach: Row-span Constrained LCL Properties.}  Instead of trying to manipulate curve-decodability to fit into an LCL property, we extend the definition of an LCL property.  That is, we return to the first, more straightforward, attempt above.  Recall that the first property (A)---that each column of $M$ is close to some degree-$\ell$ curve $u(x)$ --is easily captured by an $a$-local LCL profile.  The harder part is dealing with item (B), that no $b$ columns of $M$ should themselves lie on a degree-$\ell$ curve.  This is a global constraint on the matrix $M$, and thus is not immediately amenable to the coordinate-wise framework of \cite{lms25}.

In order to capture item (B), we introduce a \emph{row-span} constraint.  One example of such a constraint was already present in prior works~\cite{bcdz25b, lms25}.  Indeed, that framework requires that the columns of the witness matrix $M$ be distinct, which turns out to be equivalent to the property that the row-span $U$ of $M$ does not have any identically equal coordinates (that is, there should not be any $i \neq j$ so that $u_j = u_i$ for all $u \in U$).  Thus, \cite{bcdz25b,lms25} restricted the row-span of their witness matrices $M$ to the set $\mathcal{L}_\mathrm{dist}(\mathbb{F}_q^r)$ of \emph{distinct} subspaces $U$, those without any identically equal coordinates.

Our observation is that item (B) can also be captured in this way.  In \cref{lem:curve-rowspan,cor:curve-free-rowspan}, we show that item (B) is equivalent to the following: For any $B \in {A \choose b}$, there is some $u$ in the row-span $U$ of $M$ so that $u|_B$ does \emph{not} agree with any degree-$\ell$ curve $p:B \to \mathbb{F}_q$.  We define a subspace $U$ with this property to be \emph{$(\ell, A,b)$ curve-free} (\cref{def:curvefree}).  Thus, in addition to an $a$-local LCL profile $\mathbf{V}$ to enforce the condition (A), we also have a condition on the row-span---namely that the row-span $U$ of $M$ should be $(\ell, A, b)$ curve free---to enforce condition (B).

Formally, we implement this by defining the set $\mathcal{F}$ to be a set of \emph{pairs} $(\mathbf{V}, U)$, where $\mathbf{V}$ is an $r$-local profile, and $U \subset  \mathbb{F}_q^r$ is a subspace.  Then a witness matrix $M$ with columns in $C$ is ``bad'' if there is some $(\mathbf{V}, U) \in \mathcal{F}$ so that the $i$-th row of $M$ is contained in $V_i$ for all $i$; and additionally the row-span of $M$ is contained in $U$.  We say that $C$ \emph{avoids} $\mathcal{F}$ if there is no such witness matrix.

\medskip\noindent\textbf{Putting it Together.}
In order to implement the ideas above, we first re-derive the LCL framework and the connections in \cite{bcdz25b,lms25} in the row-span constrained setting.  We do this in \cref{sec:gen-lcl}.  In \cref{thm:rlc-threshold-pair,cor:rlc-threshold-family}, we extend the threshold theorem of \cite{lms25} to the row-span constrained setting, defining a ``threshold rate'' $R_\mathcal{F}$ below which random linear codes are likely to satisfy properties defined by $\mathcal{F} \subset  \mathcal{L}(\mathbb{F}_q^r)^n \times \mathcal{L}(\mathbb{F}_q^r)$ (with the additional row-span constraint) and above which they are not. 
In \cref{thm:sdc-threshold-pair,cor:sdc-threshold-family}, we extend the transfer theorem from \cite{bcdz25b} to the row-span constrained setting, showing that good subspace design codes have essentially the same ``threshold rate'' as random linear codes.
 
Next, in \cref{sec:cd-local}, we show that \emph{not} being $(\ell,\delta,a,b)$ curve-decoding is a row-span constrained $a$-local LCL property, as described above. The local profile records the agreement of the nearby codewords with an ambient degree-$\ell$ curve (property (A)), while the row-span constraint records the absence of a degree-$\ell$ codeword curve through any $b$ nearby codewords (property (B)).  This is stated formally in \cref{thm:cd-casting,cor:cd-casting}.

Finally, we apply our row-span constrained LCL theory to curve-decoding and to three random ensembles of codes---random linear codes, RS codes with random evaluation points, and Gallager's ensemble of LDPC codes---to prove \cref{thm:main}. To recap,  we know from \cite{gg25a} that subspace design codes satisfy curve-decodability; we have shown in \cref{sec:gen-lcl} that random ensembles of codes behave similarly to subspace design codes for any \emph{row-span constrained} LCL property; and we have shown in \cref{sec:cd-local} that curve-decodability is such a property.  Therefore random ensembles of codes exhibit good curve-decodability as well.  We put everything together and work out the parameters in \cref{sec:pg}; our main theorem statements are given as \cref{thm:rlc-main}, \cref{thm:rrs-main}, and \cref{thm:rldpc-main}.

\subsection{Related Works}
\label{sec:related}

\paragraph{Proximity Gaps.}
Proximity gaps (and related notions, including correlated agreement and mutual correlated agreement) arise naturally in the study of interactive protocols.  The goal in these settings is to test whether or not \emph{every} vector in a set---say, a line or a curve---is close to a code $C$.  If that code exhibits a proximity gap, then it suffices to choose a random vector in that set and use its distance from the code as a proxy for the minimum distance to $C$ over the whole set.  This connection was observed in \cite{RVW13}, and since then a long line of work has been devoted to improving the parameters~\cite{AHIV17,BKS18,BGKS20,bciks23,whir,gg25a} and establishing limitations~\cite{dg25,cs25,bchks25,hkk2026}.  See the recent survey~\cite{ABF26} for an overview of the work on proximity gaps and related notions. Prior to the work of \cite{gg25a}, the best known proximity gap results required $\delta < 1 - \sqrt{R}$, corresponding to the \emph{Johnson radius} in list-decoding.  As discussed above, the work \cite{gg25a} improved this to allow $\delta$ to approach $1-R$, the best trade-off possible.  However, as noted earlier, the results can still be improved, notably in how small the parameter $\eps$ can be.  The current best results for explicit codes---including folded RS codes, multiplicity codes, and all subspace design codes---follow from \cite{gg25a} and work under the condition given by \cref{eq:target}.  In this work we improve their results for random ensembles of codes from \cref{eq:previous} to \cref{eq:target}. 

\medskip\noindent\textbf{Local properties and subspace design codes.}
A recent line of work~\cite{mrrsw20,GLMRS21,GMRSW21,lms25} has shown that a wide variety of random ensembles of codes---including random linear codes, random LDPC codes, and Reed-Solomon codes with random evaluation points---all behave similarly with respect to ``local'' properties: properties that are characterized by the exclusion of small bad sets.  The paper \cite{lms25} introduced the LCL framework discussed above, and \cite{bcdz25b} connected this to another line of work on \emph{subspace design codes}.  The study of subspace design codes implicitly began with the introduction of Folded Reed-Solomon codes~\cite{GR08}. Over nearly two decades, it was observed that a key property of these codes, as well as of univariate multiplicity codes, is that they satisfy a \emph{subspace design property}~\cite{GK16}, which has been leveraged and generalized in subsequent works~\cite{KMRS18,Tam24,Sri25,AHS26,CZ25,bcdz25b,ggh26}.
Notably, \cite{CZ25} showed that \emph{any} such subspace design codes enjoy optimal list-decodability.  The work \cite{bcdz25b} then connected this line of work to the LCL framework discussed above, showing that in fact the same thresholds are achieved by subspace design codes.  The work \cite{gg25a} leveraged this connection to study curve-decoding, and we do the same in this work.

\subsection{Discussion and Open Problems}
\label{sec:discussion}
We conclude this section with a few open questions.
\begin{itemize}
    \item Our approach can automatically turn any improvements for subspace design codes into analogous improvements for random ensembles of codes.  Thus, the natural open question is to further improve the trade-offs between $\eps, q, \ell$ and $\eta$ for subspace design codes.
    \item Our results imply that Reed-Solomon codes with random evaluation points attain near-optimal proximity gaps, even for large $\ell$.  While it is known that there exist Reed-Solomon codes that do \emph{not} achieve optimal proximity gaps~\cite{bchks25,dg25,cs25,hkk2026}, it is not known which choices of evaluation points are ``good.''  
    Can we identify explicit Reed-Solomon codes with near-optimal proximity gaps?  This question is related to the open problem of identifying explicit Reed-Solomon codes that achieve near-optimal list-decoding.  
   \item We hope that our techniques will be useful beyond the present application.  Are there other natural properties that can be captured with the row-span constrained LCL framework, but not with the original LCL framework?
\end{itemize}
\subsection*{Acknowledgements} 
R.G. is supported by (Yael Tauman Kalai’s) grant from  Defense Advanced Research Projects Agency
(DARPA) under Contract No. HR0011-25-C-0300. Any opinions, findings and conclusions or recommendations expressed in this material are those of the author(s) and do not necessarily reflect the views of the Defense Advanced Research Projects Agency (DARPA).
V.G. is supported by a Simons Investigator award, NSF grant CCF-2211972, and ONR grant N00014-24-1-2491.
Y.S. is funded by the NSF Graduate Research Fellowship and the Stanford Graduate Fellowship. M.W. is partially funded by NSF grants CCF-2231157 and CNS-2321489.

\section{Preliminaries}
\label{sec:prelim}
We begin by introducing the basic coding-theoretic definitions we will be using.
For any two vectors $x, y$ in $\Sigma^n$ where $\Sigma$ is some alphabet, we define the fractional Hamming distance $\Delta(x, y) = \frac{1}{n} |\{i\in [n]: x_i \ne y_i\}|$ to be the fraction of coordinates where they differ.
For a set $S\subset  \Sigma^n$, we define $\Delta(x, S) = \min_{y\in S} \Delta(x,y)$
to be the fractional distance of $x$ to its closest vector in $S$.
Throughout this paper, all distances are taken to be fractional unless stated otherwise.

The two fundamental quantities associated with a code are its rate and distance.
For a code $C \subset  \Sigma^n$, we define its (relative) distance as $\delta(C) = \min_{x, y \in C, x\ne y} \Delta(x,y)$.
Moreover, its rate $R(C)$ is defined as $R(C)= \frac{1}{n} \log_{|\Sigma|}|C|$.

In this paper, we will focus on additive codes over a finite field, defined as follows:

\begin{definition}[Additive codes]\label{def: additive-code}
    Let $\FF$ be a finite field and let $\Sigma = \FF^s$ for some positive integer $s$. 
    A code $C \subset  \Sigma^n$ is said to be $\FF$-additive (or just additive when the field $\FF$ is clear from context) if $C$ is an $\FF$-linear subspace of $\Sigma^n$. 
\end{definition}

We focus on three families of random codes, defined below.  The first is a random linear code.
\begin{definition}[Random Linear Code (RLC) and Random $\mathbb{F}_q$-Additive Code]\label{def: RLC}
    A random linear code $C\subset  \FF_q^n$ of rate $R$ is given by taking $C = \{(\mathsf{Enc}_1(x),\cdots, \mathsf{Enc}_n(x))\mid x\in \FF_q^{Rn}\}$ where $\mathsf{Enc}_i : \FF_q^{Rn}\to \FF_q$ are independent uniformly random linear maps.
    \footnote{Note that, defined this way, there is some small chance that the rate $\log_q|C|/n$ of a random linear code may not actually be $R$, if there are too many linear dependencies between the maps $\mathsf{Enc}_i$.  However, it is more convenient to work with i.i.d. encoding maps, and the probability of the rate not being $R$ is negligible, so anything that holds with high probability for this model also holds with high probability for other models of random linear codes, for example the model where one chooses $C$ to be a uniformly random subspace of dimension $Rn$, or the model where $C$ is the kernel of a random matrix in $\mathbb{F}_q^{(1-R)n \times n}$.  See, e.g., \cite[Appendix A]{lms25}. }
    More generally, a random $\mathbb{F}_q$-additive code $C \subset  (\mathbb{F}_q^s)^n$ of rate $R$ is given by taking $C = \{ (\mathsf{Enc}_1(x), \ldots, \mathsf{Enc}_n(x)) \ |\  x \in \mathbb{F}_q^{Rsn} \}$ where $\mathsf{Enc}_i: \mathbb{F}_q^{Rsn} \to \mathbb{F}_q^s$ are independent uniformly random linear maps. 
\end{definition}

The next family is a random Reed-Solomon code, which is just a Reed-Solomon code where the evaluation points are chosen uniformly at random.
\begin{definition}[Random Reed-Solomon Code (RRS)]\label{def: RRS}
    A Reed-Solomon code over a field $\FF_q$ with evaluation points $\alpha_1,\cdots, \alpha_n$ is defined as follows \[\mathsf{RS}_{\FF_q}(\alpha_1,\dots, \alpha_n, k) = \{(f(\alpha_1),\dots, f(\alpha_n))\mid f\in \FF_q[X], \deg f<k\}\]
    A random Reed-Solomon code of rate $R$ is one in which $k=Rn$ and $\alpha_i$ are independent uniformly random elements of $\FF_q$.
\end{definition}

Finally, we introduce our last family of codes, random Low-Density Parity-Check (LDPC) codes.
In this paper, by ``random LDPC codes'' we mean \emph{Gallager's ensemble} of random LDPC codes~\cite{gallager}. In general, a rate $R$, $s$-sparse LDPC code can be described as follows: 
Pick a bipartite graph $G=(L\sqcup R, E)$ with $|L|=n$, $|R|=(1-R)n$ and right degree $s$. Each edge $e\in E$ has a weight $w_e\in \FF_q$, the code $C$ is given as follows: \[C = \Bigl\{c\in \FF_q^L: \sum_{u\in N(v)}w_{(u,v)}c_u=0 \text{ for all }v\in R\Bigr\}.\]
Above, $N(v)$ denotes the neighborhood of $v$ in $G$.

To define a random LDPC code, we thus just need to describe the underlying bipartite graph and the edge weights $w_e$.  Gallager's ensemble~\cite{gallager} defines an ensemble of random graphs (and hence random LDPC codes) as follows.\footnote{Technically, this way of generalizing to larger alphabets is slightly different than \cite{gallager}.  Our definition coincides with \cite{gallager} for $q=2$; for larger $q$, \cite{gallager} chooses weights $w_e$ to all be $1$, while we choose them uniformly in $\mathbb{F}_q^*$; this matches the ensemble studied in~\cite{mrrsw20}.}

\begin{definition}[Random Low Density Parity Check Code (RLDPC)]\label{def: RLDPC}
    Let $t=(1-R)s$, and let $G_i = (V, W_i, E_i)$ for $i=1,\cdots, t$ be uniformly random $(1,s)$ regular bipartite graphs with a shared left vertex set and disjoint right vertex sets $W_i$ of size $n/s$. The graph $G$ is then defined as $\bigsqcup_i G_i$; that is, $G = (V,W,E)$, where $W = \bigsqcup_i W_i$ and $E = \bigsqcup_i E_i$.  The edge weights $w_e$ are sampled as independent uniformly random elements of $\FF_q^\ast$. A \emph{Random LDPC Code} $C$ of rate $R$ is then the LDPC code derived from $G$ and the weights $\{w_e\}$, as described above.
\end{definition}

\subsection{Proximity Gaps and (Mutual) Correlated Agreement}
\label{sec:prelim-pg-mca}
Proximity Gaps are related to two stronger notions, \emph{Correlated Agreement} (CA) and \emph{Mutual Correlated Agreement} (MCA).  Our results actually establish both CA and MCA as well as proximity gaps.  

Informally, we say that a code $C$ has $(\ell, \delta, \eps, \gamma$)-correlated agreement if, for all degree-$\ell$ curves $u(x) = \sum_{i=0}^\ell u_i x^i$ with $u_i \in \Sigma^n$, the following holds: Either at most an $\eps$ fraction of the points on $u(x)$ are $\delta$-close to $C$; or there is some ``codeword curve'' $c(x) = \sum_{i=0}^\ell c_i x^i$ with $c_i \in C$ so that $u(x)$ agrees identically with $c(x)$ on a set $S \subset  [n]$ of at least $\delta n/(1 -\gamma)$ coordinates.  Formally, we have the following definition.
\begin{definition}[Correlated Agreement (CA)]\label{def: curve CA}
    An $\FF$-additive code $C\subset  \Sigma^n$ is said to have $(\ell, \delta, \eps, \gamma)$ correlated agreement if for all $u_0, u_1, \cdots, u_\ell \in \Sigma^n$ and $\delta'\le \delta$ 
    \[\mb{P}_{\alpha\in \FF}\Bigl[ \Delta\bigl(\sum_{j=0}^\ell u_j \alpha^j, C\bigr)\le \delta'\Bigr]>\eps\implies
    \exists c_0, c_1,\cdots c_{\ell} \in C:
    \Bigl|\Bigl\{i : \bigvee_{j=0}^\ell  c_{j, i}\ne u_{j, i}\Bigr\}\Bigr| \le\frac{\delta' n}{1-\gamma} . \] 
\end{definition}
   In particular, CA implies that for all $\alpha \in \FF$, we have that $\Delta\left(\sum_{j=0}^\ell u_j \alpha^j, \sum_{j=0}^\ell c_j \alpha^j\right)\le {\delta'}/({1-\gamma})$. Thus, $(\ell, \delta, \eps, \gamma)$ correlated agreement implies an $(\ell, \delta, \eps, \gamma)$proximity gap.

Next, we define Mutual Correlated Agreement (MCA).  Informally, we say that an $\mathbb{F}$-additive code $C$ has $(\ell, \delta, \eps)$ MCA if for all degree-$\ell$ curves $u(x)$ in $\Sigma^n$,
the following holds.  For all sets $S \subset  [n]$ of size at least $(1 - \delta)n$, and for at least an $1 - \eps$ fraction of $\alpha \in \mathbb{F}$, if there is some $c \in C$ that agrees with $u(\alpha)$ on $S$, then there is some ``codeword curve'' $c(x) = \sum_{i=0}^\ell c_i x^i$ for $c_i \in C$ so that $c(x)$ and $u(x)$ agree identically on $S$.  Formally, we have the following definition, which captures the contrapositive of the intuition above.

\begin{definition}[Mutual Correlated Agreement (MCA)]\label{def: curve MCA}
    An $\FF$-additive code $C\subset  \Sigma^n$ is said to have $(\ell, \delta, \eps)$ mutual correlated agreement if for all $u_0, u_1, \cdots, u_\ell \in \Sigma^n$ and $\delta'\le \delta$ 
    \[\mb{P}_{\alpha\in\FF}\left(\exists S \subset  [n], c\in C, j\in\{0, 1, \dots, \ell\}: |S|\ge (1-\delta')n, c_{\mid S}=\sum_{j=0}^\ell \alpha^j u_{j\mid S},  c'_{\mid S}\ne u_{j\mid S} \forall c'\in C \right)\le \eps  . \]
\end{definition}
We note that $(\ell, \delta, \eps)$ MCA implies $(\ell, \delta, \eps, \gamma=0)$ CA and hence $(\ell, \delta, \eps, \gamma=0)$ PG.  Indeed, suppose that $C$ satisfies MCA, and suppose there is some codeword $c \in C$ that is $\delta'$-close to $u(x)$ on at least $\eps q$ points, call them $A \subset  \mathbb{F}$.  MCA implies that there is a set $B \subset  \mathbb{F}$ of size at least $(1 - \eps)q$ so that the event inside the $\mb{P}_{\alpha}[\cdot]$ statement holds.  Then there must be at least one point $\alpha \in A \cap B$.  Letting $S \subset [n]$ be the set of coordinates (of size at least $(1 - \delta')n$) on which $u(\alpha)$ agrees with $c$, we conclude from the definition of MCA that there is some ``codeword curve'' $c(x)$ so that $c(x)$ and $u(x)$ identically agree on $S$, establishing CA.

Finally, we formally define curve-decodability, which we informally defined in \cref{sec:intro} (see \cref{fig:curve-decoding}).
The following definition was given in~\cite{gg25a}, and the case of $\ell=1$ was independently defined in \cite{hab24} as \textit{collinearity of proximates}. We adopt the convention from \cite{gg25a}.

\begin{restatable}[Curve-Decodability]{definition}{cddef}\label{def: curve-decodability}
    An additive code $C\subset  \Sigma^n$ is $(\ell, \delta, a, b)$ curve-decodable if for every $u_0, u_1,\cdots, u_\ell\in \Sigma^n$,  all functions $f: \FF_q \to C$, whenever the set 
    \[ A=\Bigl\{\alpha \in \FF : \Delta\bigl(\smudef, f(\alpha)\bigr) \le  \delta\Bigr\}\] has at least $a$ elements,  there exist $c_0, c_1,\cdots, c_\ell \in C$ such that 
    \[\Bigl|\Bigl\{\alpha \in A  :  f(\alpha)=\smcdef \Bigr\}\Bigr|\ge b \ .  \]
\end{restatable}

\cite{gg25a} showed a transformation to boost results about the curve-decodability of a code assuming that the underlying code is also list-decodable.

\begin{theorem}[{\cite[Theorem 3.6]{gg25a}}]\label{thm:cd+ld}
Suppose an $\FF_q$-additive code $C\subset  (\mb{F}_q^s)^n$ is 
\begin{equation}
\left(\delta\left(1+\frac{\ell}{b-\ell}\right), L\right)\text{ list-decodable \quad and \quad } \left(\ell, \delta, a, b\right)\text{ curve-decodable}.
\end{equation}
Then $C$ is also $\left(\ell, \delta, (T-1)L+a, T\right)$ curve-decodable for any $T$.
\end{theorem}

The key motivation to define curve-decodability is due to the following implication for correlated and mutual correlated agreement.

\begin{theorem}[Correlated agreement from curve-decodability, {\cite[Theorem 3.4]{gg25a}}]\label{thm: CA from only curves}
    Let $C\subset  \Sigma^n$ be an $\FF_q$-additive code. Let $\ell, t>\ell, a\in \mathbb N$ and $\delta \in (0,1)$. Suppose that $\mathcal C$ is $(\ell,\delta, a, t)$ curve-decodable.
    Then $C$ has $(\ell,\delta, {a}/{q},\ell/t)$ correlated agreement.
\end{theorem}

\begin{theorem}[Mutual correlated agreement from curve-decodability, {\cite[Theorem 3.5]{gg25a}}]\label{thm: appendix MCA from curves}
    Let $C\subset  \Sigma^n$ be an $\FF_q$-additive code. Let $\ell, a\in \mathbb N$ and $\delta \in (0,1)$. Suppose that $\mathcal C$ is $(\ell, \delta, a, \ell n + 1)$ curve-decodable. Then $C$ has $(\ell, \delta, {a}/{q})$ mutual correlated agreement.
\end{theorem}

\begin{corollary}\label{cor: main-ca-mca-gg}
If an $\FF_q$-additive code $C\subset  (\mb{F}_q^s)^n$ is \begin{equation}
\left(\delta\left(1+\frac{\ell}{b-\ell}\right), L\right)\text{ list-decodable \quad and \quad } \left(\ell, \delta, a, b\right)\text{ curve-decodable}
\end{equation} 
then $C$ has: \begin{enumerate}
    \item for any $t\ge 1$, $(\ell, \delta, (t-1)L+a, t)$ curve-decodability,
    \item for any $m >1$, $\left(\ell, \delta, \frac{\ell m L + a}{q}, \frac{1}{m}\right)$ correlated agreement and proximity gap,
    \item $(\ell, \delta, \frac{\ell n L+a}{q})$ mutual correlated agreement.
\end{enumerate}
\end{corollary}
\begin{proof}
    The first statement is simply \cref{thm:cd+ld}. For the second, we set $T=\ell m$ and apply \cref{thm: CA from only curves}. For the third, we set $T = \ell n +1$ and apply \cref{thm: appendix MCA from curves}. 
\end{proof}

Enabled by the above reductions, to prove an additive code $C$ has correlated agreement, it suffices to prove list-decodability and curve-decodability of the code. In this work, and related works, we always work in the regimes where the codes in consideration are already known to be list-decodable with good parameters. Thus, the key innovation lies in establishing curve-decodability.

\subsection{Subspace-Design Codes}
We introduce notation and results from \cite{gg25a} that we will use.

\begin{definition}[Subspace-Design Code]\label{def: Subspace-designs}
    For any function $\tau: \mathbb N\rightarrow \R_{\le 1}$, an $\FF_q$-additive code $C\subset  (\FF_q^s)^n$ is said to be a $\tau$-\textit{subspace design} code if for every $r\in \mathbb N$, and every $\F_q$-linear subspace $A$ of $C$ of dimension at most $r$, the following holds: 
    \[\frac{1}{n}\sum\limits_{i=1}^n \dim A_i\le \dim(A)\cdot \tau(r) \]
    where $A_i = \{a\in A\mid a_i=0\}$.
\end{definition}

\begin{definition}\label{def:sdc}
For an additive code $C\subset \left(\mb{F}_q^s\right)^n$ and dimension $r$, let
\begin{equation}
\sigma(r)\coloneqq \max\left\{\mb{E}_{i\in [n]}\left[\frac{\dim V_i}{\dim V}\right]:V\subset C, \dim V=r\right\}
\end{equation}
where $V_i = \{x\in V:x_i=0\}$. 
In the message space with encoding maps $\on{Enc}_i:\mb{F}_q^k\to\mb{F}_q^s$, we have
\begin{equation}\label{eq:sdc-message}
	\sigma(r)\coloneqq \max\left\{\mb{E}_{i\in [n]}\left[\frac{\dim (A\cap \ker(\on{Enc}_i))}{\dim A}\right]:A\subset \mb{F}_q^k, \dim A=r\right\}
\end{equation}
Following the notation of \cite{gg25a}, we can define the monotone envelope $\tau(r)\coloneqq \max\{\sigma(1), \dots, \sigma(r)\}$.  Observe that with this choice, $C$ is by definition a $\tau$-subspace design code.
\end{definition}
\begin{remark}[Comparison to \cite{bcdz25b}]
The paper \cite{bcdz25b} uses slightly different language; they say that $C$ is \emph{$r$-subspace designable if}, for every $d\in [r]$
	\begin{equation}
	\sigma(d)\le R+\frac{1}{nd}.
	\end{equation}
We recall that $\sigma (r)\ge R-r/sn$ always. Thus, having the subspace design property with ``good'' parameters means that we should have $\tau(d)\le R+\eta$ for some ``small'' $\eta>0$. \cite{bcdz25b} proved that 
 random $\mathbb{F}_q$-additive codes $C \subset  (\mathbb{F}_q^s)^n$ are $r$-subspace designable whenever $s=\Omega(rn)$.
\end{remark}

Finally, we state a theorem from \cite{gg25a} which shows that any subspace-design code is curve-decodable.
\begin{theorem}[{\cite[Theorem 4.7]{gg25a}}]\label{thm:sdc-good-cd}
Fix any integers $a, \ell$ and $\eta>0$, and let $d=\lceil(\ell+1)/\eta\rceil$. Then for any $\tau$-subspace design code $C\subset(\mb{F}_q^s)^n$,  $C$ is 
\begin{equation}
\left(\ell, 1-\tau(d)-\eta, a, \frac{\eta  a}{d+\eta}\right)\text{ curve-decodable}.
\end{equation}
\end{theorem}

\subsection{Notations and Conventions}
\label{sec:notations}
We let $[n]\coloneqq \{1, \ldots, n\}$ with the convention $[0]=\emptyset$. Let $\binom{n}{\le k}$ and $\binom{n}{\ge k}$ denote the sum of binomial coefficients $\binom{n}{i}$ for $i\le k$ and $i\ge k$, respectively. Let $\binom{S}{k}$ denote the set of $k$-element subsets of $S$. $\binom{S}{\le k}$ and $\binom{S}{\ge k}$ denote the set of subsets of $S$ with size at most $k$, and at least $k$, respectively. For a finite set $S$, let $\mb{E}_{s\in S}$ denote the expectation over a uniform random element $s$ of $S$.

For a matrix $A\in\mb{F}^{m\times n}$, let $A_{i\star}\in\mb{F}^n$ be the $i$-th row of $A$ and let $A_{\star j}\in\mb{F}^m$ be the $j$-th column of $A$. Given a vector space $V$, let $\mc{L}(V)$ be the set of (linear) subspaces of $V$.

We also adopt standard notation from asymptotic analysis: as $x\to\infty$, we write  $f(x)=o(g(x))$ if $f(x)/g(x)\to 0$;  $f(x)=O(g(x))$ or $g(x)=\Omega(f(x))$ if there exists a finite, positive constant $C$ such that $f(x)\le Cg(x)$ for all sufficiently large $x$.

We use $h_q:[0,1] \to [0,1]$ and $h_q^{-1}$ to denote the $q$-ary entropy function and its inverse, i.e.
\[ h_q(x) := x \log_q(q-1) + x \log_q\left(\frac{1}{x}\right) + (1 - x) \log_q\left( \frac{1}{1 - x}\right).\]

\section{Row-span Constrained Local Properties}
\label{sec:gen-lcl}
\subsection{Generalizing Local Profiles}
We begin by introducing some notation from \cite{lms25,bcdz25b}, and explaining how we generalize it to our row-span constrained versions.  As described in \cref{sec:overview}, a \emph{local profile} gives a set of coordinate-wise linear constraints that a ``bad'' witness matrix $M$ might satisfy.
\begin{definition}[$r$-Local Profile, \cite{lms25}]\label{def: local profile}
Let $\FF_q$ be a finite field of $q$ elements and let $n$ and $r$ be positive integers. A tuple
    $\mathbf{V} = (V_1, \cdots, V_n)\in \mathcal L(\FF_q^r)^n$ is called a $r$-local profile.
\end{definition}

\noindent Next, we define what it means for a code to \emph{contain} a local profile (with a particular row-span).
\begin{definition}[Containing a local profile, \cite{lms25}]\label{def: containing a local profile}
    A code $C \subset  \FF_q^n$ is said to \emph{contain} a local profile $\mathbf{V}$ with span $U\in \mathcal L(\FF_q^r)$ (often written ``$C$ \emph{contains} $(\mathbf{V}, U)$'') if there exists an $M\in \FF_q^{n\times r}$ such that:
    \begin{itemize}
        \item each column of $M$ is contained in $C$; 
        \item the $i$-th row of $M$ is contained in $V_i$; and
        \item the row-span of $M$ is $U$.
    \end{itemize}
We call such an $M$ a \emph{witness matrix}, and columns of $M$ \emph{witness columns}.
\end{definition}

In \cite{lms25}, an LCL property is given by a collection $\mathcal{F}$ of $r$-local profiles $\mathbf{V}$.  
 Our key generalization is to also record the corresponding row-spans in the family.  That is, $\mathcal{F}$ is no longer a collection of profiles $\mathbf{V}$; now it is a collection of \emph{pairs} $(\mathbf{V}, U)$, where $U$ constrains the row-span of a witness matrix.

\begin{definition}[Row-span Constrained LCL Property]\label{def:rsclp}
A \emph{row-span constrained $r$-local LCL family} $\mc{F}$ is a collection of \emph{$r$-local pairs} $(\mbf{V},U)$ where $\mbf{V}=(V_1, \dots, V_n)\in\mc{L}(\mb{F}_q^r)^n$ is an $r$-local profile and $U\in\mc{L}(\mb{F}_q^r)$ is the row-span constraint of the witness matrix. 

A code $C$ \emph{contains} $\mc{F}$ if it contains some pair $(\mbf{V},U)\in\mc{F}$.
It \emph{avoids} $\mathcal F$ otherwise.
 We say that the \emph{row-span constrained $r$-local property} given by $\mathcal{F}$ is the property of \emph{avoiding} $\mathcal{F}$.\footnote{We note that other sources, like \cite{lms25}, define an $r$-local property to be the property of \emph{containing} $\mathcal{F}$.  Since we are concerned with the property of avoiding a family $\mathcal{F}$, we define it this way.}
\end{definition}
\begin{remark}
\cref{def:rsclp} generalizes the definition of an LCL property in~\cite{lms25}, in the sense that they restrict the columns of their witness matrices $M$ to be distinct; as they observe, this is the same as restricting the row-span of the witness matrices to lie in $\mathcal{L}_{\mathrm{dist}}(\mathbb{F}_q^r)$, where $ \mc{L}_{\on{dist}}(\mb{F}_q^r)$ denotes all subspaces $U$ of $\mb{F}_q^r$ so that no two distinct coordinates $i, j\in [r]$ are identically equal on $U$. 
\end{remark}

Next, we define what it means for an $\mathbb{F}_q$-additive code with alphabet $\Sigma = \mathbb{F}_q^s$ to contain $(\mathbf{V}, U)$.

\begin{definition}[Folded containment]\label{def:folded-profile}
Let $C\subset(\mb{F}_q^s)^n$ be an $\mathbb{F}_q$-additive code. For an
$r$-local profile $\mbf{V}=(V_1,\ldots,V_n) \in (\mathcal{L}(\mathbb{F}_q))^n$, define the \emph{duplicated profile}
$\mbf{V}^{\otimes s} \in (\mathcal{L}(\mathbb{F}_q))^{sn}$ by
\begin{equation}\label{eq:tensor-profile}
	{V}^{\otimes s}_{(i,j)}=V_i
\end{equation}
for any $i\in[n]$ and $j\in[s]$.
We say that $C$ contains $(\mbf{V},U)$ if the ``unfolded'' version 
$C \subset \mb{F}_q^{sn}$
contains $(\mbf{V}^{\otimes s},U)$. 
Here, by the ``unfolded code,'' we mean to embed $(\mathbb{F}_q^s)^n$ in $\mathbb{F}_q^{sn}$ in the natural way.

For a family $\mathcal F$ of pairs, define
\[
	\mathcal F^{\otimes s}
	=
	\{(\mbf{V}^{\otimes s},U):(\mbf{V},U)\in\mathcal F\}.
\]
Thus $\mathbb{F}_q$-additive code $C\subset(\mb{F}_q^{s})^n$ contains $\mc{F}$ if and only if the ``unfolded'' version $C\subset\mb{F}_q^{sn}$ contains
$\mathcal F^{\otimes s}$.
\end{definition}
\begin{remark}\label{rem:unfolding-confusion}
Note that $\mc{F}^{\otimes s}$ does not encode the same property as $\mc{F}$ applied to unfolded codes. For example, proximity is not preserved: being far as an $\mathbb{F}_q$-additive code over $\Sigma = \mathbb{F}_q^s$ does not mean it is far when we unfold, since we can change one coordinate per block of $s$. The onus is on the property proving step, to show for any additive code, the desired property profiles take the form \cref{eq:tensor-profile}
for every $s$.
\end{remark}

Now, we define the potential function and the threshold of row-span constrained $r$-local LCL properties identically as \cite{lms25}, up to slightly modified notational conventions and normalization by $n$. We emphasize that these definitions are code-independent.

\begin{definition}[Potential and Threshold, \cite{lms25}]
For an $r$-local pair $(\mbf{V}, U)$, and a real number $R$, define the \emph{potential function}
\begin{equation}
    \phi_\mbf{V}(U,R)
	\coloneqq
	(R-1)\dim U+\mb{E}_{i\in[n]}\dim(U\cap V_i).
\end{equation}
For $W\subsetneq U$, define
\begin{equation}
	R_\mbf{V}(U,W)
	\coloneqq
	1-
	\mathbb E_{i\in[n]}
	\left[
	\frac{
	\dim(U\cap V_i)-\dim(W\cap V_i)
	}{
	\dim U-\dim W
	}
	\right].
\end{equation}
Define \emph{threshold} of the $r$-local pair $(\mbf{V}, U)$ as
\begin{equation}
	R_\mbf{V}(U)
	\coloneqq
	\max_{W\subsetneq U}R_\mbf{V}(U,W).
\end{equation}
For a row-span constrained $r$-local LCL property, define
\begin{equation}\label{eq:Rf}
	R_{\mc{F}}
	\coloneqq
	\min_{(\mbf{V},U)\in\mc{F}}R_\mbf{V}(U).
\end{equation}
\end{definition}
The only difference between the definition above and that in \cite{lms25} is the fact that the minimum in \eqref{eq:Rf} is over profile-subspace pairs $(\mathbf{V},U) \in \mathcal{F}$ rather than over profiles $\mathbf{V} \in \mathcal{F}$.

Next, we unpack the definitions to interpret the threshold as a root of the potential function.

\begin{lemma}[Threshold interpretation]
\label{lem:threshold-interpretation}
For every $r$-local profile $\mbf{V}$, $W\subsetneq U\in\mc{L}(\mb{F}_q^r)$, and $R\in\mb{R}$,
\begin{equation}\label{eq:threshold-potential}
	\phi_{\mbf{V}}(U,R)-\phi_{\mbf{V}}(W,R)
	=
	(\dim U-\dim W)(R-R_{\mbf{V}}(U,W)).
\end{equation}
Let $\mc{F}$ be a row-span constrained $r$-local LCL family
and let $\eta\ge 0$. We observe that
\begin{enumerate}
	\item If $R\le R_{\mc{F}}-\eta$, then for every $(\mbf{V},U)\in\mc{F}$
	there exists $W\subsetneq U$ such that
	\begin{equation}
		\phi_{\mbf{V}}(U,R)-\phi_{\mbf{V}}(W,R)\le -\eta.
	\end{equation}
	\item If $R\ge R_{\mc{F}}+\eta$, then there exists $(\mbf{V},U)\in\mc{F}$
	such that for every $W\subsetneq U$,
	\begin{equation}
		\phi_{\mbf{V}}(U,R)-\phi_{\mbf{V}}(W,R)\ge \eta.
	\end{equation}
\end{enumerate}
\end{lemma}

\begin{proof}
To prove \cref{eq:threshold-potential}, we expand the left-hand side to obtain
\begin{equation}
\begin{aligned}
\phi_{\mbf{V}}(U,R)-\phi_{\mbf{V}}(W,R)
&=
(R-1)(\dim U-\dim W) \\
&\quad+
\mb{E}_{i\in[n]}
\left[
\dim(U\cap V_i)-\dim(W\cap V_i)
\right] \\
&=
(\dim U-\dim W)
\left(
R-1+
\mb{E}_{i\in[n]}
\left[
\frac{\dim(U\cap V_i)-\dim(W\cap V_i)}
{\dim U-\dim W}
\right]
\right) \\
&=
(\dim U-\dim W)(R-R_{\mbf{V}}(U,W)).
\end{aligned}
\end{equation}
For (1), fix $(\mbf{V},U)\in\mc{F}$. Since
\begin{equation}
	R_{\mc{F}}\le R_{\mbf{V}}(U)=\max_{W\subsetneq U}R_{\mbf{V}}(U,W),
\end{equation}
there exists $W\subsetneq U$ such that
$R\le R_{\mbf{V}}(U,W)-\eta$.
Hence
\begin{equation}
\phi_{\mbf{V}}(U,R)-\phi_{\mbf{V}}(W,R)
=
(\dim U-\dim W)(R-R_{\mbf{V}}(U,W)) \le
-\eta(\dim U-\dim W)
\le -\eta.
\end{equation}
For (2), choose $(\mbf{V},U)\in\mc{F}$ with
$R_{\mbf{V}}(U)=R_{\mc{F}}$. Then for every $W\subsetneq U$,
\begin{equation}
	R\ge R_{\mc{F}}+\eta=R_{\mbf{V}}(U)+\eta\ge R_{\mbf{V}}(U,W)+\eta.
\end{equation}
Therefore
\begin{equation}
\phi_{\mbf{V}}(U,R)-\phi_{\mbf{V}}(W,R)
=
(\dim U-\dim W)(R-R_{\mbf{V}}(U,W)) \ge
\eta(\dim U-\dim W)
\ge \eta,
\end{equation}
as desired.
\end{proof}
We also see that by our definition of profile containment for additive (rather than linear) codes, the threshold rate $R_V$ is preserved for profiles taking tensor form in \cref{eq:tensor-profile} as $s$ changes. 
Intuitively, this is because $(U, \mathbf{V}, W)$ are all independent of $s$, so potential function
$
 \phi_{\mathbf{V}^{\otimes s}}(U, R)=\phi_{\mathbf{V}}(U, R)$.
Then, $R_{\mathbf{V}^{\otimes s}}(U, W)$ is simply its root. We verify this directly.

\begin{proposition}\label{prop:unfolding}
Let $(\mbf{V}, U)$ be any $r$-local pair on plain codes in $\mb{F}_q^n$ and let $W\subsetneq U$.
Then, for every $s\in\mb{N}$
\begin{equation}
R_{V^{\otimes s}}(U,W)=R_V(U,W).
\end{equation}
Consequently, for any row-span constrained local LCL family, $R_{\mc{F}}=R_{\mc{F}^{\otimes s}}$ is independent of $s$.
\end{proposition}

\begin{proof}
Fix any $W\subsetneq U$. By definition,
the length-$sn$ profile $\mathbf{V}^{\otimes s}$ is given by
$V^{\otimes s}_{(i,j)}\coloneqq V_i$ and indexed by $(i, j)\in[n]\times [s]$. Then, we see that 
\begin{align*}
R_{\mathbf{V}^{\otimes s}}(U,W)
& =
1-
\mb{E}_{(i,j)\in[n]\times[s]}
\left[
\frac{
\dim(U\cap {V}^{\otimes s}_{(i,j)})
-
\dim(W\cap {V}^{\otimes s}_{(i,j)})
}{
\dim U-\dim W
}
\right].
\\ & =
1-
\mb{E}_{i\in[n]}
\left[
\frac{
\dim(U\cap V_i)
-
\dim(W\cap V_i)
}{
\dim U-\dim W
}
\right] 
\\ & =
R_\mathbf{V}(U,W).
\end{align*}
Maximizing over $U$ on both sides gives
$\max_{W\subsetneq U}R_{\mathbf{V}^{\otimes s}}(U,W)=\max_{W\subsetneq U}R_\mathbf{V}(U,W)$, as desired.
\end{proof}
\subsection{Behavior Under Quotients}
We next record how local profile containment behaves under quotient maps.

\begin{remark}[Comparison to \cite{ggh26}]  In \cite{ggh26}, the spaces $V_i, U$ can lie in an arbitrary vector space over $\mb{F}_q$.  In contrast, our definition of $r$-local pairs require $V_i, U\subset\mb{F}_q^r$ and is thus more concrete. However, this is at the expense of a less elegant quotient behavior: since our profiles are coordinate-dependent, a quotient by
$W\subsetneq U$ is implemented by choosing a surjective linear map
\begin{equation}
\pi:\mb{F}_q^r\to\mb{F}_q^{r'}\quad\text{where}\quad \ker \pi =W.
\end{equation}
This is the same as quotienting by $W$ and choosing coordinates
on the quotient.
\end{remark}

\begin{definition}[Quotient profile]
Let $\mbf{V}=(V_1,\ldots,V_n)$ be an $r$-local pair and let
$\pi:\mb{F}_q^r\to\mb{F}_q^{r'}$ be a surjective linear map, so $r'\le r$. Define the $r'$-local profile
\begin{equation}
	\pi\mbf{V}\coloneqq(\pi V_1,\ldots,\pi V_n).
\end{equation}
For an $r$-local pair $(\mbf{V}, U)$ and surjective linear map $\pi$, if $W=\ker \pi\subsetneq U\subset \mb{F}_q^r$, then we say $(\pi\mbf{V}, \pi U)$ is a quotient pair of $(\mbf{V}, U)$. Then, the condition $W\subsetneq U$ is equivalent to $\pi U\ne 0$.
\end{definition}

While the locality is changed in a quotient pair which in particular means it is no longer in $\mc{F}$, we will see next that the potential, threshold, and code-containment are well-behaved.

\begin{lemma}[Quotient Containment]\label{lem: quotienting preserves containment}
    Suppose code $C$ contains $r$-local pair $(\mbf{V}, U)$ and $\pi: \FF_q^r\to \FF_q^{r'}$ is a surjective linear map such that $\pi U\ne 0$. Then, $C$ contains $(\pi\mbf{V}, \pi U)$.
\end{lemma}
\begin{proof}
Let $M\in\mb{F}_q^{n\times r}$ witness that $C$ contains $(\mbf{V},U)$. Write
$\pi$ as an $r'\times r$ matrix, and define
\begin{equation}
	M'\coloneqq M\pi^{\mathsf T}\in\mb{F}_q^{n\times r'}.
\end{equation}
Every column of $M'$ is a linear combination of the columns of $M$, hence lies
in $C$. For every $i\in[n]$,
\begin{equation}
	M'_{i\star}=\pi(M_{i\star})\in \pi V_i.
\end{equation}
Finally, the row-span of $M'$ is $\pi U$,
so $M'$ witnesses that $C$ contains $(\pi\mbf{V},\pi U)$.
\end{proof}

We remark that this definition of code containment works for all folding $s$.

\begin{lemma}[Quotient Potential]\label{lem:quotienting}
Let $\pi:\mb{F}_q^r\to\mb{F}_q^{r'}$ be a surjective linear map with
kernel $W\subsetneq U\in\mc{L}(\mb{F}_q^r)$. Then, for any $R\in\mb{R}$,
\begin{equation}
	\phi_{\pi\mbf{V}}(\pi U,R)
	=
	\phi_{\mbf{V}}(U,R)-\phi_{\mbf{V}}(W,R).
\end{equation}
\end{lemma}
\begin{proof}
As $W\subset U$, $\dim (\pi U)=\dim U-\dim W$.
For any $i\in [n]$, we compute via rank-nullity that
\begin{equation}
	\dim(\pi U\cap \pi V_i)
	=
	\dim(\pi(U\cap (V_i+W)))=\dim (U\cap V_i+W)-\dim W=\dim (U\cap V_i)-\dim  (W\cap V_i)
\end{equation}
by repeatedly noting $\ker\pi=W\subset U$. Now, substituting into the definition of the potential gives
\begin{equation}
\begin{aligned}
\phi_{\pi\mbf{V}}(\pi U,R)
&=
(R-1)\dim (\pi U)+\mb{E}_{i\in[n]}\dim(\pi U\cap \pi V_i) \\
&=
(R-1)(\dim U-\dim W)+
\mb{E}_{i\in[n]}
\left[
\dim(U\cap V_i)-\dim(W\cap V_i)
\right] \\
&=
\phi_{\mbf{V}}(U,R)-\phi_{\mbf{V}}(W,R),
\end{aligned}
\end{equation}
as desired.
\end{proof}

Finally, we present a useful corollary of the low-rate code profile containment via the quotient characterization. We emphasize that this is code-independent, and reduces the low-rate directions of the threshold proofs of various codes into simple code-specific counting. 
\begin{lemma}[Maximal Quotient]
\label{lem:key-factored}
If a code $C$ contains an $r$-local pair $(\mbf{V}, U)$ and $R\le R_\mbf{V}(U)-\eta$ for some $R, \eta\ge 0$, then there exists a surjective linear map $\pi: \mb{F}_q^r\to\mb{F}_q^{r'}$ such that 
\begin{itemize}
    \item $W\coloneqq \ker \pi\subsetneq U$,
    \item $C$ contains $r'$-local pair $(\pi\mbf{V}, \pi U)$,
    \item $\phi_{\pi\mbf{V}}(\pi U, R)\le -\eta\dim \pi U$.
\end{itemize}
Moreover, we can choose $\pi$ such that $\phi_{\pi\mbf{V}}(U', R)\le-\eta\dim U'$ for all nonzero subspaces $U'\subset  \pi U$.
\end{lemma}
\begin{proof}From assumption, we know that there exists subspace $W\subsetneq U$ such that $	R\le R_{\mbf{V}}(U,W)-\eta$. By \cref{lem:threshold-interpretation}, we know that
\begin{equation}
\phi_{\mbf{V}}(U,R)-\phi_{\mbf{V}}(W,R)
=(\dim U-\dim W)(R-R_{\mbf{V}}(U,W))
\le\phi_{\mbf{V}}(U,R)-\phi_{\mbf{V}}(W,R)
\le -\eta.
\end{equation}
Choose $W\subsetneq U$ maximal under inclusion among those satisfying %
\begin{equation}
	\phi_{\mbf{V}}(W,R)\ge \phi_{\mbf{V}}(U,R)+\eta.
\end{equation}
Let $\pi:\mb{F}_q^r\to\mb{F}_q^{r'}$ be any surjective linear map with kernel $W\subsetneq U$.
Because $C$ contains $(\mbf{V},U)$, \cref{lem: quotienting preserves containment}
implies that $C$ contains $(\pi\mbf{V},\pi U)$. By \cref{lem:quotienting},
\begin{equation}
	\phi_{\pi\mbf{V}}(\pi U,R)
	=
	\phi_{\mbf{V}}(U,R)-\phi_{\mbf{V}}(W,R)
	\le
	-\eta\dim (\pi U).
\end{equation}
This proves the three bullet points.

For the last conclusion, pick any nonzero subspace $U'$ of $\pi U$.
Let $\widetilde U\coloneqq \pi^{-1}(U')\cap U$.
Then
\begin{equation}
	W\subsetneq \widetilde{U}\subsetneq U.
\end{equation}
The containment is strict on the left because $U'\ne 0$, and strict on the
right because $U'\subsetneq \pi U$.
By \cref{lem:quotienting} applied to $\widetilde U$ whose image under $\pi$ is $U'$, we have that
\begin{equation}\label{eq:moreover}
	\phi_{\pi\mbf{V}}(U',R)
	=
	\phi_{\mbf{V}}(\widetilde U,R)-\phi_{\mbf{V}}(W, R) = \phi_{\mbf{V}}(\widetilde U,R+\eta)-\phi_{\mbf{V}}(W, R+\eta) -\eta(\dim (\tilde U)-\dim W).
\end{equation}
By maximality of $W$, the strictly larger subspace $\widetilde{U}$ must satisfy
\begin{equation}
	\phi_{\mbf{V}}(\widetilde U,R+\eta)\leq\phi_{\mbf{V}}(W,R+\eta)
\end{equation} and $\dim \tilde U -\dim W  = \dim U'$.
Therefore, \cref{eq:moreover} is at most $-\eta \dim U'$, proving the last statement.
\end{proof}

\subsection{Thresholds for Random Ensembles of Codes}
In this section, we prove that $R_{\mbf{V}}(U)$ is the threshold rate for various random ensembles of codes $C$ to contain $(\mbf{V}, U)$. Then, we can lift it to row-span constrained local LCL families $\mc{F}$ via a union bound over all pairs in $\mc{F}$.

\subsubsection{Random Linear Codes}
\begin{lemma}\label{lem:profile-size}
    Let $(\mathbf V, U)$ be an $r$-local pair, let $\mc{M}_\mathbf{V}(U)$ be the set of matrices $M\in\mb{F}_q^{n\times r}$ with row-span $U$ that satisfy $\mbf{V}$, and let $\mc{M}_\mbf{V}^\star(U)$ be those with row-span contained in $U$. For any $R\in \mb{R}$,
    \begin{equation}
    |\mc{M}_\mbf{V}(U)|\le   |\mc{M}_\mbf{V}^\star(U)| =  q^{n((1-R)\cdot \dim U + \phi_\mbf{V}(U,R))}.
    \end{equation}
\end{lemma}
\begin{proof}
    Let $\mathbf V = (V_1,\cdots,V_n)$. Now, the $i$-th row of any such matrix in $\mc{M}_\mbf{V}^\star(U)$ is $V_i\cap U$. There are $q^{\dim (V_i\cap U)}$ choices for the row, so \[|\mc{M}_\mbf{V}^\star(U)|=q^{\sum_{i=1}^n \dim (V_i\cap U)}=q^{n((1-R)\cdot \dim U+\phi_\mbf{V}(U,R))}.\]
Trivially, $|\mc{M}_\mbf{V}(U)|\le   |\mc{M}_\mbf{V}^\star(U)|$.
\end{proof}

\begin{theorem}[Fixed-Pair Random Linear Code Threshold]\label{thm:rlc-threshold-pair}
Let $(\mbf{V}, U)$ be an $r$-local pair and let $C\subset\mb{F}_q^n$ be a random linear code of rate $R$.
\begin{enumerate}
    \item If $R\le R_{\mbf{V}}(U) -\eta$ for some $\eta>0$, then 
    \begin{equation}
    \mb{P}\left(C\text{ contains }(\mbf{V}, U)\right)\le q^{-\eta n}
    \end{equation}
    \item  If $R\ge R_{\mbf{V}}(U) +\eta$ for some $\eta>0$, then 
    \begin{equation}
    \mb{P}\left(C\text{ contains }(\mbf{V}, U)\right)\ge 1- q^{r^2-\eta n}
    \end{equation}%
\end{enumerate}
\end{theorem}

The low-rate direction (1) will essentially follow from \cref{lem:key-factored} and a few simple lemmas that we state first.
\begin{lemma}[{\cite[Lemma 4.1]{lms25}}]\label{fact:lms-obs}
If $C\subset\mb{F}_q^n$ is an RLC of rate $R$, then for any $M\in\mb{F}_q^{n\times r}$
\begin{equation}\label{eq:rank-bd}
\mb{P}\left(M_{\star j}\subset C\, \forall\, j\in [r]\right) = q^{(R-1)n\cdot \on{rank}(M)}
\end{equation}
\end{lemma}

The idea behind \cref{fact:lms-obs} is to view $C$ as the kernel of $(1-R)n$ many linear constraints in $\mb{F}_q^n$ which column-span of $M$ needs to satisfy. Each of $\on{rank}(M)$ basis vectors satisfies them with probability $q^{-(1-R)n}$. We note that this is the only RLC-specific ingredient we use. 

Combining with \cref{lem:profile-size}, we obtain the following corollary, which intuitively explains why $\phi_\mbf{V}(U, R)$ characterizes the containment of $(\mbf{V}, U)$ in a random linear code of rate $C$.
\begin{corollary}
\label{cor:expected-profile-rlc}
The expected number of witnesses matrices $M$ for the containment of an $r$-local pair $(\mathbf{V}, U)$ in a random linear code $C$ with rate $R\in [0,1]$ is at most $q^{n\phi_\mbf{V}(U,R)}$.
\end{corollary}
\begin{proof}
By \cref{lem:profile-size}, the number of matrices with row-span
$U$ that satisfy $\mbf{V}$ is at most
\begin{equation}\label{eq:expected-profile-rlc-1}
q^{n((1-R)\cdot \dim U +\phi_{\mbf{V}}(U,R))}
\end{equation}
For any such matrix, it is contained in $C$ with probability at most
$q^{(R-1)n\dim U}$
by \cref{fact:lms-obs}. Combining with \cref{eq:expected-profile-rlc-1} gives the desired bound.
\end{proof}

With these lemmas, we deduce \cref{thm:rlc-threshold-pair}.
\begin{proof}[Proof of \cref{thm:rlc-threshold-pair}]
For the low-rate case (1), if $C$ contains $(\mbf{V},U)$, then by \cref{lem:key-factored}, there is a surjective linear map
$\pi=\pi_{(\mbf{V},U)}:\mb{F}_q^r\to\mb{F}_q^{r'}$
such that $C$ contains $(\pi\mbf{V},\pi U)$ and
$\phi_{\pi\mbf{V}}(\pi U,R)\le -\eta$. 
By a union bound and \cref{cor:expected-profile-rlc},
\begin{equation}
	\mb{P}\left(C\text{ contains }(\mbf{V},U)\right)\le \mb{P}\left(C\text{ contains }(\pi\mbf{V},\pi U)\right)
\le
q^{n\phi_{\pi\mbf{V}}(\pi U,R)}
\le q^{-\eta n}.
\end{equation}

For the high-rate case (2), we know that $R\ge R_{\mbf{V}}(U,W)+\eta$ for every $W\subsetneq U$. Then, $\phi_\mbf{V}(U, R)\ge \phi_\mbf{V}(W, R)+\eta$ for all $W\subsetneq U$.
Let random variable $X_W$ be the number of witnesses matrices $M$ for the containment by $C$ of an $r$-local pair $(\mathbf{V}, W)$. By \cref{cor:expected-profile-rlc}, we know that
\begin{equation}\label{eq:EXW}
\mb{E}[X_W]\le q^{n\phi_{\mbf{V}}(W, R)}
\end{equation}

Now, let us fix any code $C$ of rate $R\in [0, 1]$ and consider matrices $\mc{M}_\mbf{V}^\star(U)$ that satisfies $\mbf{V}$ and has row-span contained in $U$. By \cref{lem:profile-size}, we know that
\begin{equation}
 |\mc{M}_\mbf{V}^\star(U)| =  q^{n((1-R)\cdot \dim U + \phi_\mbf{V}(U,R))}.
\end{equation}
Now, for any realization of $C$, requiring each column of $M$ to lie in $C$ imposes at most $(1-R)n\dim U$ independent constraints, since $C$ has codimension $(1-R)n$ and the row-span is contained in $U$. Therefore, the number of $M\in\mc{M}_\mbf{V}^\star(U)$ contained in $C$ is 
\begin{equation}
\sum_{W\in\mc{L}(U)} X_W \ge  |\mc{M}_\mbf{V}^\star(U)| q^{(R-1)n\dim U} = q^{n\phi_{\mbf{V}}(U, R)}
\end{equation}
deterministically for every realization $C$. Now, $C$ contains $(\mbf{V}, U)$ if and only if $X_U>0$, so
\begin{equation}
\begin{aligned}
\mb{P}(C \text{ does not contain }(\mbf{V}, U)) & =\mb{P}(X_U=0)
\\ & \le\mb{P}\left(
\sum_{W\subsetneq U} X_W\ge q^{n\phi_{\mbf{V}}(U, R)}\right)
\\ & \le \frac{\sum_{W\subsetneq U}\mb{E}[X_W]}{q^{n\phi_{\mbf{V}}(U, R)}}
\\ & \le \sum_{W\subsetneq U} q^{n\phi_{\mbf{V}}(W, R)-n\phi_{\mbf{V}}(U, R)}
\\ & \le |\mc{L}(U)|q^{-\eta n}
\end{aligned}
\end{equation}
where we apply Markov's inequality and \cref{eq:EXW}. Now, the conclusion follows $|\mc{L}(U)|\le q^{r^2}$.
\end{proof}

We now lift \cref{thm:rlc-threshold-pair} to the whole row-span constrained LCL family $\mc{F}$.

\begin{corollary}[Random Linear Code Threshold]
\label{cor:rlc-threshold-family}
Let $\mc{F}$ be a row-span constrained $r$-local property and let $C\subset\mb{F}_q^n$ be a random linear code of rate $R$.
\begin{enumerate}
    \item If $R\le R_{\mc{F}}-\eta$ for some $\eta>0$, then 
    \begin{equation}
    \mb{P}\left(C\text{ contains }\mc{F}\right)\le \left|\mc{F}\right|q^{-\eta n}
    \end{equation}
    \item  If $R\ge R_{\mc{F}} +\eta$ for some $\eta>0$, then 
    \begin{equation}
    \mb{P}\left(C\text{ contains }\mc{F}\right)\ge 1- q^{r^2-\eta n}
    \end{equation}
\end{enumerate}
\end{corollary}
\begin{proof}
For (1), note that by a union bound and \cref{thm:rlc-threshold-pair}(1), we have that
\begin{equation}
    \mb{P}\left(C\text{ contains }\mc{F}\right)\le    \sum_{(\mbf{V}, U)\in\mc{F}}\mb{P}\left(C\text{ contains } (\mathbf{V}, U)\right) \le|\mc{F}|q^{-\eta n}
\end{equation}
as $R\le R_{\mc{F}}\le R_{\mbf{V}}(U)$ for any $(\mbf{V}, U)\in\mc{F}$. This proves (1).

For (2), let $(\mbf{V}, U)$ be the minimizing pair such that $R_\mc{F}=R_{\mbf{V}}(U)$. Then, $R\ge R_\mbf{V}(U)+\eta$, so 
\begin{equation}
\mb{P}\left(C\text{ contains }(\mbf{V}, U)\right)\ge 1- q^{r^2-\eta n}
\end{equation}
by \cref{thm:rlc-threshold-pair}(2).  If $C$ contains $(\mathbf{V},U)$, then $C$ contains $\mc{F}$ by definition, which establishes the corollary.
\end{proof}

Therefore, when $|\mc{F}|=q^{o(n)}$\footnote{Some previous literature, e.g. \cite{lms25}, call such $\mc{F}$ \emph{reasonable}}, we know that $R_\mc{F}$ is the threshold for RLCs for large $n$ and $q$: RLCs of rate less than $R_\mc{F}$ avoid $\mc{F}$ with high probability; RLCs of rate greater than $R_\mc{F}$ contain $\mc{F}$ with high probability.

\subsubsection{Random Reed-Solomon Codes}
We next record the analogue of the low-rate local threshold \cref{cor:rlc-threshold-family}(1) for random Reed-Solomon code: if the
rate is below the threshold of a row-span constrained LCL family, the random Reed-Solomon
code is unlikely to contain the family. We use the following result from \cite{lms25} as a black box.
\begin{lemma}[{\cite[Proposition 6.1]{lms25}}]\label{lem: RRS-threshold}
    Let $n\le q$ with $q$ a prime power, and let $r\in \mathbb N$. Let $\mathbf V = (V_1,\cdots, V_n)$ be an $r$-local profile, $U\in \mathcal L(\FF_q^r)$, and let $0\leq R \leq 1$ such that for all $U'\in \mathcal L(U)$, \[\phi_\mbf{V}(U', R)\leq -\eps\cdot \dim U'.\]
    Then, with probability taken over uniformly random $\alpha_1,\dots, \alpha_n \in \FF_q$, we have 
    \[\mb{P}\left(C=\mathsf{RS}_{\FF_q}(\alpha_1,\dots, \alpha_n, k) \text{~contains $V$ with span $U \neq \{0\}$}\right)\leq (2^r-1) \left(\frac{(4r)^{4r}k}{\eps q}\right)^{\eps n/2r} \ .\]
\end{lemma}

\begin{corollary}[Random Reed-Solomon Code Threshold: Low Rate Direction]
\label{cor:rrs-threshold-family}
Let $\mc{F}$ be a row-span constrained $r$-local property and let $C\subset\mb{F}_q^n$ be a random Reed-Solomon code of rate $R$.
    If $R\le R_{\mc{F}}-\eta$ for some $\eta>0$, then 
    \begin{equation}\label{eqn: RRS-non-containment}
    \mb{P}\left(C\text{ contains }\mc{F}\right)\le \left|\mc{F}\right|(2^r-1) \left(\frac{(4r)^{4r}Rn}{\eta q}\right)^{\eta n/2r}
    \end{equation}
\end{corollary}
\begin{proof}
    We first assume that due to \cref{lem:key-factored}, for every $(\mathbf V, U)\in \mc{F}$, there exists a surjective linear map $\pi: \FF_q^r \mapsto \FF_q^{r'}$ such that if $C$ contains $(\mathbf V, U)$, then it also contains $(\pi \mathbf V, \pi U)$, and for all subspaces $U'\subset  \pi U$ $\phi_{\pi \mathbf V}(U',R)<-\eta\dim U'$.
    Now, \cref{eqn: RRS-non-containment} follows directly from \cref{lem: RRS-threshold}.
\end{proof}

\subsubsection{Random LDPCs (Gallager's Ensemble)}\label{subsec: random LDPC}
We next record the analogue of the low-rate local threshold \cref{cor:rlc-threshold-family}(1) and \cref{cor:rrs-threshold-family} for random LDPC codes drawn from Gallager's Ensemble (\cref{def: RLDPC}).

There is one important difference from random linear codes.  For RLCs, the
probability that a fixed matrix $M$ is contained in the code is exactly
$q^{(R-1)n\on{rank}(M)}$ by \cref{fact:lms-obs}.  For LDPC codes $C$, this exact identity is
false.  Instead, \cite{mrrsw20} proves that, after conditioning on the usual
good-distance event defined for $C=\mathsf{RLDPC}(n, q, s, R)$ as
\begin{equation}
\label{eq:ldpc-good}
    \mathsf{Good}
    \coloneqq
    \left\{
        \Delta(C)\ge \frac{1}{2}h_q^{-1}(1-R)
    \right\},
\end{equation}
then $C$ contains any fixed matrix with
essentially the same probability as a random linear code, up to a factor
$q^{\eps n}$.  This is the only LDPC-specific input we need.

We first recall the lemma from \cite{mrrsw20} saying that $\mathsf{Good}$ happens with high probability.

\begin{lemma}[{\cite[Theorem 2.14]{mrrsw20}}]
\label{lem:rldpc-good-distance}
For any finite field $\FF_q$ and $R\in[0,1-1/q)$, there exists
$s_0=s_0(R,q)$ such that for every $s>s_0$, with probability over $C=\mathsf{RLDPC}(n,q,s,R)$
\begin{equation}
    \mb{P}
    \left(\mathsf{Good}\right)\coloneqq \mb{P}
    \left(
        \Delta(C)\ge \frac{1}{2}h_q^{-1}(1-R)
    \right)
    \ge 1-o_{n\to\infty}(1).
\end{equation}
\end{lemma}

The next lemma is the fixed-pair containment estimate for row-span constrained local profiles.
This lemma follows \cite{mrrsw20}, but we restate it in our framework and provide a proof in \cref{app: RLDPC} for completeness.

\begin{restatable}[RLDPCs containing local profiles]{lemma}{rldpclocalprofiles}\label{lem: RLDPC-threshold}
    For any $\eps, R \in (0,1)$, positive integer $r$, finite field $\FF_q$, there exists an $s_0=s_0(\eps, r, R, q)$ such that the following holds for all odd $s>s_0$ and sufficiently large $n$.
   Let $\mbf{V}=(V_1,\ldots,V_n)$ be an $r$-local profile and let
$U\in\mc{L}(\FF_q^r)$. For $C=\mathsf{RLDPC}(n,q,s,R)$,
\begin{equation}
\mb{P}\left(
    C\text{ contains }(\mathbf{V},U)
\middle|
    \mathsf{Good}
\right)
\le
q^{n\phi_{\mathbf{V}}(U,R)+\eps n} \quad\text{where}\quad \mathsf{Good}
    \coloneqq
    \left\{
        \Delta(C)\ge \frac{1}{2}h_q^{-1}(1-R)
    \right\}.
\end{equation}
\end{restatable}%

We next obtain the thresholds for RLPDC codes. The  changes from the RLC counterparts are:
\begin{itemize}
    \item We need to condition on the good-distance event $\mathsf{Good}$ and then account for the error term $\mb{P}(\mathsf{Good}^c)=o_{n\to\infty}(1)$ from \cref{lem:rldpc-good-distance}.
    \item We apply the LDPC-specific replacement \cref{lem: RLDPC-threshold} instead of the RLC first moment bound \cref{fact:lms-obs}, with an additional $q^{\eta n/2}$ slack in the exponent. 
\end{itemize}

Other than these changes, we follow \cref{thm:rlc-threshold-pair}: the low-rate direction again uses the quotient lemma to reduce containment below threshold to containment of a quotient pair with negative potential; the high-rate direction again uses the deterministic lower bound on the number of contained matrices with row-span contained in $U$, followed by Markov's inequality to show that lower-dimensional row-spans cannot account for all such matrices. 

\begin{theorem}[Fixed-Pair Random LDPC Threshold]
\label{thm:rldpc-threshold-pair}
Let $(\mbf{V},U)$ be an $r$-local pair and let
$C=\mathsf{RLDPC}(n,q,s,R)$ be a random LDPC code from Gallager's ensemble. Define the event 
\begin{equation}
    \mathsf{Good}
    \coloneqq
    \left\{
        \Delta(C)\ge \frac{1}{2}h_q^{-1}(1-R)
    \right\}.
\end{equation}

For every $R,\eta\in(0,1)$, positive integer $r$, and finite field $\FF_q$,
there exists $s_0=s_0(\eta,r,R,q)$ such that for every odd $s>s_0$, the
following holds for all sufficiently large $n$:
\begin{enumerate}
    \item If $R\le R_{\mbf{V}}(U)-\eta$, then $\mb{P}\left(C\text{ contains }(\mbf{V},U)\middle|\mathsf{Good}\right)
        \le
     q^{-\eta n/2}$, and so
    \begin{equation}
        \mb{P}\left(C\text{ contains }(\mbf{V},U)\right)
        \le
     q^{-\eta n/2}+o_{n\to\infty}(1).
    \end{equation}

    \item If $R\ge R_{\mbf{V}}(U)+\eta$, then $\mb{P}\left(C\text{ contains }(\mbf{V},U)\middle|\mathsf{Good}\right)
        \ge
        1-q^{r^2-\eta n/2}$, and so
    \begin{equation}
        \mb{P}\left(C\text{ contains }(\mbf{V},U)\right)
        \ge
        1-q^{r^2-\eta n/2}-o_{n\to\infty}(1).
    \end{equation}
\end{enumerate}
\end{theorem}
\begin{proof}
By \cref{lem:rldpc-good-distance}, after increasing $s_0$ if necessary,
$\mb{P}(\mathsf{Good})\ge 1-o_{n\to\infty}(1)$, where
\begin{equation}
    \mathsf{Good}
    \coloneqq
    \left\{
        \Delta(C)\ge \frac{1}{2}h_q^{-1}(1-R)
    \right\}.
\end{equation}

We first prove the low-rate case (1). Suppose
$R\le R_{\mbf{V}}(U)-\eta$. By \cref{lem:key-factored}, there exists a
surjective linear map $\pi:\FF_q^r\to\FF_q^{r'}$ such that if $C$ contains
$(\mbf{V},U)$, then $C$ contains $(\pi\mbf{V},\pi U)$, and
$\phi_{\pi\mbf{V}}(\pi U,R)\le -\eta$. Therefore, by
\cref{lem: RLDPC-threshold} with parameter $\eta/2$,
\begin{equation}
    \mb{P}\left( C\text{ contains }(\mbf{V},U)\middle|\mathsf{Good} \right)
\le
\mb{P}\left(C\text{ contains }(\pi\mbf{V},\pi U)\middle|\mathsf{Good}
\right)\le
q^{n\phi_{\pi\mbf{V}}(\pi U,R)+\eta n/2} \le
q^{-\eta n/2}.
\end{equation}
Adding $\mb{P}(\mathsf{Good}^c)=o_{n\to\infty}(1)$ gives the low-rate case (1).

For the high-rate case (2), we know that
$R\ge R_{\mbf{V}}(U,W)+\eta$ for every $W\subsetneq U$. By \cref{lem:threshold-interpretation}, $\phi_{\mbf{V}}(U,R)\ge \phi_{\mbf{V}}(W,R)+\eta$.
Let random variable $X_W$ be the number of witnesses matrices $M$ for the containment by $C$ of an $r$-local pair $(\mathbf{V}, W)$. By
\cref{lem: RLDPC-threshold} with parameter $\eps=\eta/2$, 
\begin{equation}\label{eq:EXW-good}
    \mb{E}\left[X_W\mid \mathsf{Good}\right]
    \le
    q^{n\phi_{\mbf{V}}(W,R)+\eta n/2},
\end{equation}
for every $W\in\mc{L}(U)$.
Next, as in the RLC case, we let
$\mc{M}_{\mbf{V}}^\star(U)$ be the vector space of matrices
$M\in\FF_q^{n\times r}$ satisfying $\mbf{V}$ and having row-span contained
in $U$. By \cref{lem:profile-size},
\begin{equation}
    |\mc{M}_{\mbf{V}}^\star(U)|
    =
    q^{n((1-R)\dim U+\phi_{\mbf{V}}(U,R))}.
\end{equation}
Now, for any realization of $C$, requiring each column of $M$ to lie in $C$ imposes at most $(1-R)n\dim U$ independent constraints, since $C$ has codimension $(1-R)n$ and the row-span is contained in $U$. Therefore, the number of $M\in\mc{M}_\mbf{V}^\star(U)$ contained in $C$ is 
\begin{equation}
    \sum_{W\in\mc{L}(U)}X_W
    \ge
    |\mc{M}_{\mbf{V}}^\star(U)|q^{(R-1)n\dim U}
    =
    q^{n\phi_{\mbf{V}}(U,R)}.
\end{equation}
deterministically for every realization $C$.
Now, $C$ contains $(\mbf{V}, U)$ if and only if $X_U>0$. Therefore,
conditioned on $\mathsf{Good}$, we apply Markov's inequality and \cref{eq:EXW-good} as before to obtain
\begin{equation}
\begin{aligned}
\mb{P}\left(
    C\text{ does not contain }(\mbf{V},U)
    \middle|
    \mathsf{Good}
\right)
& = \mb{P}\left(
   X_U=0
    \middle|
    \mathsf{Good}
\right) 
\\
&\le
\mb{P}\left(
    \sum_{W\subsetneq U}X_W
    \ge q^{n\phi_{\mbf{V}}(U,R)}
   \middle| 
    \mathsf{Good}
\right) \\
&\le
\frac{
    \sum_{W\subsetneq U}\mb{E}\left[X_W\mid \mathsf{Good}\right]
}{
    q^{n\phi_{\mbf{V}}(U,R)}
} \\
&\le
\sum_{W\subsetneq U}
q^{n(\phi_{\mbf{V}}(W,R)-\phi_{\mbf{V}}(U,R))+\eta n/2} \\
&\le
|\mc{L}(U)|q^{-\eta n/2} \\
&\le
q^{r^2-\eta n/2}.
\end{aligned}
\end{equation}
as $|\mc{L}(U)|\le q^{r^2}$.
Adding $\mb{P}(\mathsf{Good}^c)=o_{n\to\infty}(1)$ gives
the high-rate case (2).
\end{proof}

\begin{remark}
We will only use the statement (1) from \cref{thm:rldpc-threshold-pair} later in this paper, but we include the statement and proof of (2) for completeness.
This should not be confused with a converse to the local-property transfer theorem for RLDPC codes (\cite[Theorem 1.9]{mrrsw20}), which is false at that level of generality, as observed in \cite{mrrsw20}; our argument uses the additional linear coordinate-wise profile structure of row-span constrained LCL properties.
\end{remark}

As before, we now lift \cref{thm:rldpc-threshold-pair} to the whole row-span constrained LCL family $\mc{F}$.

\begin{corollary}[Random LDPC Code Threshold]
\label{cor:rldpc-threshold-family}
Let $\mc{F}$ be a row-span constrained $r$-local property and let
$C=\mathsf{RLDPC}(n,q,s,R)$ be a random LDPC code from Gallager's ensemble (\cref{def: RLDPC}).
For every $R,\eta\in(0,1)$, positive integer $r$, and finite field $\FF_q$,
there exists $s_0=s_0(\eta,r,R,q)$ such that for every odd $s>s_0$, the
following holds for all sufficiently large $n$:
\begin{enumerate}
    \item If $R\le R_{\mc{F}}-\eta$, then 
    \begin{equation}
    \mb{P}\left(C\text{ contains }\mc{F}\right)\le \left|\mc{F}\right|q^{-\eta n/2}+o_{n\to\infty}(1)
    \end{equation}
    \item  If $R\ge R_{\mc{F}} +\eta$, then 
    \begin{equation}
    \mb{P}\left(C\text{ contains }\mc{F}\right)\ge 1- q^{r^2-\eta n/2}-o_{n\to\infty}(1)
    \end{equation}
\end{enumerate}
\end{corollary}
\begin{proof}
For (1), note that by a union bound and \cref{thm:rldpc-threshold-pair}(1), we have that
\begin{equation}
    \mb{P}\left(C\text{ contains }\mathcal{F}\right)\le  \mb{P}\left(\mathsf{Good}^c\right)+  \sum_{(\mathbf{V}, U)\in\mathcal{F}}\mb{P}\left(C\text{ contains }(\mathbf{V}, U)\middle|\mathsf{Good}\right) \le o_{n\to\infty}(1) + |\mc{F}|q^{-\eta n/2},
\end{equation}
as $R\le R_{\mc{F}}\le R_{\mbf{V}}(U)$ for any $(\mbf{V}, U)\in\mc{F}$. This proves (1).

For (2), let $(\mbf{V}, U)$ be the minimizing pair such that $R_\mc{F}=R_{\mbf{V}}(U)$. Then, $R\ge R_\mbf{V}(U)+\eta$, so 
\begin{equation}
\mb{P}\left(C\text{ contains }(\mbf{V}, U)\right)\ge 1- q^{r^2-\eta n/2}-o_{n\to\infty}(1)
\end{equation}
by \cref{thm:rldpc-threshold-pair}(2).  If $C$ contains $(\mathbf{V}, U)$, then $C$ contains $\mc{F}$ by definition, which establishes the corollary.
\end{proof}

\subsection{Threshold for Subspace Design Codes}

Informally, for ordinary $r$-local properties where every row-span $U\in \mc{L}_{\on{dist}}(\mb{F}_q^r)$ is permitted for each profile $\mbf{V}$, \cite[Theorems 3.1, 4.2]{bcdz25b} shows that the threshold of containing $\mbf{V}$ is dictated by $\tau(r)$ versus $R_\mbf{V}$. For good subspace design codes, like folded Reed-Solomon codes, we have $\tau(r)\approx R$, which means that good subspace design codes and random linear codes share the same threshold. 
We generalize these results to row-span constrained local properties. We emphasize that the transfer works at the level of any pair $(\mbf{V}, U)$ via the following theorem.
\begin{theorem}[Fixed-Pair Subspace Design Code Threshold]\label{thm:sdc-threshold-pair}
Let $(\mbf{V}, U)$ be an $r$-local pair. Then
\begin{enumerate}
	\item Let $C\subset(\mb{F}_q^s)^n$ be an additive $\tau$-subspace design code. If $\tau(r) < R_\mbf{V}(U)$, then $C$ does not contain $(\mbf{V}, U)$.
	\item For any $\eta>0$, there exists 
    a $\tau$-subspace design additive code $C\subset(\mb{F}_q^s)^n$ with $\tau(r)\le R_{\mbf{V}}(U)+\eta$ that contains $(\mbf{V}, U)$.
\end{enumerate}
\end{theorem}

The statement (1) in \cref{thm:sdc-threshold-pair} essentially follows from \cref{lem:key-factored} and the subspace design condition.
\begin{proof}[Proof of \cref{thm:sdc-threshold-pair}(1)]
Suppose $\tau (r)\le R_\mbf{V}(U)-\eta$, then by \cref{lem:key-factored} there exists surjective linear map $\pi:\mb{F}_q^r\to\mb{F}_q^{r'}$ such that $\phi_{\pi\mbf{V}}(\pi U, \tau(r))<0$ and $C$ contains $(\pi\mbf{V}, \pi U)$. The latter implies the existence of a witness matrix $M$.  Let the columns of the witness matrix $M$ for $(\pi\mbf{V},\pi U)$ be
$c_j=\on{Enc}(a_j)$
and define the message subspace
$A\coloneqq \on{span}(a_j)$.
Since the encoder is injective and the witness matrix has row-span $\pi U$, $d\coloneqq \dim A=\dim(\pi U).$

For each coordinate $i\in[n]$, the image $\on{Enc}_i(A) \subset  \mathbb{F}_q^s$ is
the column span of the $i$-th ``block'' of the witness matrix. The rows of
that block lie in $\pi U\cap \pi V_i$, hence
\begin{equation}
\dim(\pi U\cap \pi V_i)\ge 	\dim \on{Enc}_i(A) = \dim A - \dim (A\cap \ker(\on{Enc}_i)).
\end{equation}
where the equality follows from rank-nullity.
Averaging over $i$ and applying the subspace design condition \cref{def: Subspace-designs} with $d=\dim A\le r$, we obtain that
\begin{equation}
	d-\mb{E}_{i\in[n]}\dim(\pi U\cap \pi V_i)\le \mb{E}_{i\in[n]}\dim(A\cap\ker(\on{Enc}_i))
	\le
	\tau(r)d.
\end{equation}
Rearranging gives
$\phi_{\pi\mbf{V}}(\pi U,\tau(r))\ge 0$,
a contradiction. Thus, $C$ does not contain
$(\mbf{V},U)$.
\end{proof}

For the statement (2) in \cref{thm:sdc-threshold-pair}, we use the following theorem from \cite{bcdz25b} that shows RLCs are good subspace design codes when viewed as an $\mathbb{F}_q$-additive code with alphabet $\mathbb{F}_q^s$.

\begin{theorem}[{\cite[Theorem 4.1]{bcdz25b}}]\label{thm:rlc-are-sdc}
Fix any $\eta>0$. A random $\mathbb{F}_q$-additive code $C\subset(\mb{F}_q^s)^n$ with rate $R=k/sn$ satisfies
\begin{equation}
\mb{P}\left(\tau(r)\le R+\eta\right)\ge 1-rq^{-\eta srn+3(r^2+1)n}.
\end{equation}
\end{theorem}
\begin{proof}[Proof of \cref{thm:sdc-threshold-pair}(2)]
Fix $\eta>0$ and let $R= R_{\mbf{V}}(U)+\eta/2$. For sufficiently large $s$, we know by \cref{thm:rlc-are-sdc} that, for a random $\mathbb{F}_q$-additive code $C \subset (\mathbb{F}_q^s)^n$ of rate $R$,
\begin{equation}
\mb{P}\left(\tau(r)\le R_{\mbf{V}}(U)+\eta=R+\frac{\eta}{2}\right)\ge 1-o_{s\to\infty}(1),
\end{equation}

On the other hand, $R= R_{\mbf{V}}(U)+\eta/2=R_{\mbf{V}^{\otimes s}}(U)+\eta/2$ by \cref{prop:unfolding}.  By \cref{thm:rlc-threshold-pair},
\begin{equation}
\mb{P}(C\text{ contains }(\mbf{V}, U)) = \mb{P}(\text{unfolded }C\text{ contains }(\mbf{V}^{\otimes s}, U)) \ge 1-q^{r^2-\eta sn/2}=1-o_{s\to\infty}(1).
\end{equation}
Therefore, we can pick sufficiently large $s$ such that with probability $2/3$, $C$ contains $(\mbf{V}, U)$ and is a $\tau$-subspace design code with $\tau(r)\le R_\mbf{V}(U)+\eta$. This proves (2).
\end{proof}
\begin{remark}
Compared to \cite[Theorem 4.2]{bcdz25b}, we take the subspace design parameter $d$ in $\tau(d)$ to be $d=r$ for simplicity.
\end{remark}

We can obtain analogous threshold statements for $R_{\mc{F}}$ by taking the minimum over $\mc{F}$.
\begin{corollary}[Subspace Design Code Threshold]
\label{cor:sdc-threshold-family}
Let $\mc{F}$ be any row-span constrained $r$-local LCL family. Then
\begin{enumerate}
	\item Let $C\subset(\mb{F}_q^s)^n$ be an additive $\tau$-subspace design code. If $\tau(r) < R_\mc{F}$, then $C$ avoids $\mathcal{F}$. 
	\item For any $\eta>0$, there exists 
    a $\tau$-subspace design additive code $C\subset(\mb{F}_q^s)^n$ with $\tau(r)\le R_\mc{F}+\eta$ that contains $\mc{F}$.
\end{enumerate}
\end{corollary}
\begin{proof}
For (1), note that for any $(\mbf{V}, U)\in \mc{F}$,  we have that $\tau(r)<R_\mc{F}\le R_{\mbf{V}}(U)$, so $C$ does not contain $(\mbf{V}, U)$ by \cref{thm:sdc-threshold-pair}. This holds for any pair in $\mc{F}$, so $C$ avoids $\mc{F}$.

For (2), choose a pair $(\mbf{V},U)\in\mc{F}$ attaining the minimum, so $R_\mc{F}=R_{\mbf{V}}(U)$. By \cref{thm:sdc-threshold-pair}(2), for any $\eta>0$, there exists a
$\tau$-subspace design additive code $C\subset(\mb{F}_q^s)^n$ such that
\begin{equation}
	\tau(r)\le R_{\mbf{V}}(U)+\eta=R_{\mc{F}}+\eta
\end{equation}
and $C$ contains $(\mbf{V},U)$. Since $(\mbf{V},U)\in\mc{F}$, this means that
$C$ contains $\mc{F}$.
\end{proof}

\section{Curve-Decoding as a Row-Span Constrained LCL Property}

\label{sec:cd-local}
Recall the definition of curve-decoding:
\cddef*

As observed in \cite{gg25a} and discussed in \cref{sec:overview}, curve-decoding does not seem to fit immediately into the LCL framework.  However, we will show that it does fit into a \emph{row-span constrained} LCL framework, via \cref{def:rsclp}. Before we describe the correct row-span constraint, we define
the degree-$\ell$ Reed-Solomon evaluation subspace on $B$ by
\[
	\on{RS}_{\ell}(B)
	\coloneqq
	\left\{
	(p(\alpha))_{\alpha\in B}:
	p\in\mb{F}_q[x],\ \deg p\le \ell
	\right\}
	\subset \mb{F}_q^B.
\]
Thus $\on{RS}_{\ell}(B)$ is the space of scalar degree-$\ell$
evaluation patterns on the parameter set $B$.

\begin{definition}[Curve-free Subspaces]\label{def:curvefree}
Let $A\subset\mb{F}_q$. A subspace $U\in\mc{L}(\mb{F}_q^A)$ is
\emph{$(\ell, A,b)$ curve-free} if
$\pi_B(U)\not\subset \on{RS}_{\ell}(B)$ for every $B\in\binom{A}{b}$, where $\pi_B:\mb{F}_q^A\to\mb{F}_q^B$ denote coordinate restriction to $B$.
\end{definition}

For this definition to be non-trivial, we assume $\ell+1\le b\le|A|$. Now, we observe the following lemma, which translates $(\ell, A, b)$ curve-freeness of the row-span of a matrix to the condition that no $b$ columns lie on a degree-$\ell$ curve. 

\begin{lemma}
\label{lem:curve-rowspan}
Let $B\subset A\subset\mb{F}_q$ and let $N$ be a positive integer.\footnote{Think of $N$ as either the block length $N = n$ of a linear code $C \subset \mathbb{F}_q^n$; or the block length $N = sn$ of an ``unfolded'' $\mathbb{F}_q$-additive code $C \subset (\mathbb{F}_q^s)^n$, thought of as $C \subset \mathbb{F}_q^{sn}$.} Fix a matrix $M\in\mb{F}_q^{N\times A}$  with row-span $U\in\mc{L}(\mb{F}_q^A)$. Then, $\pi_B(U)\subset \on{RS}_{\ell}(B)$ if and only if there exists a degree-$\ell$ curve $g:\mb{F}_q\to\mb{F}_q^N$ such that $g(\alpha)=M_{\star\alpha}$ for every $\alpha\in B$.
\end{lemma}

\begin{proof}
Suppose such a curve $g$ exists. For each $i\in[N]$, the
restricted row $(M_{i\alpha})_{\alpha\in B}=(g(\alpha)_m)_{\alpha\in B}$. Since any $v\in \pi_B(U)$ is a linear combination of these
restricted rows, we observe that $v\in \on{RS}_{\ell}(B)$ as witnessed by the corresponding linear combination of degree-$\ell$ polynomials $g_i$.

Conversely, suppose $\pi_B(U)\subset \on{RS}_{\ell}(B)$. Then, for each $i\in[N]$, the restricted row
$(M_{i\alpha})_{\alpha\in B}\in \on{RS}_{\ell}(B)$. Hence
there exists a scalar polynomial $g_i(x)$ of degree at most $\ell$ such that
$g_i(\alpha)=M_{i\alpha}$ for every $\alpha\in B$. Collecting these scalar
polynomials gives a degree-$\ell$ curve
$g(x)=(g_1(x),\ldots,g_N(x))\in\mb{F}_q^N$ gives
$g(\alpha)=M_{\star\alpha}$ for every $\alpha\in B$, proving the second condition.
\end{proof}

Observe that latter condition is exactly the kind of codeword condition we wish to impose on columns of witness matrices for the failure of $(\ell, \delta, a, b)$ curve-decoding. \cref{lem:curve-rowspan} gives an equivalent row-span constraint. We make this more formal below.

\begin{corollary}
\label{cor:curve-free-rowspan}
Let $A\subset\mb{F}_q$ and let $M\in\mb{F}_q^{N\times A}$ be a
matrix with row-span $U$ whose columns lie in a linear code $C\subset\mathbb{F}_q^N$. If $\ell<b\le |A|$, then $U$ is
$(\ell,A,b)$ curve-free if and only if no $b$ columns of $M$ lie on a
degree-$\ell$ ``codeword curve.''  (That is, if and only if there is no curve $c(x) = \sum_{i=0}^\ell c_i x^i$ for $c_i \in C$ and no set $B \in {A \choose b}$ so that $c(\alpha) = M_{\star \alpha}$ for all $\alpha \in B$.)
\end{corollary}

\begin{proof}
By \cref{lem:curve-rowspan}, for each $B\in\binom{A}{b}$, the condition
$\pi_B(U)\subset \on{RS}_{\ell}(B)$ is equivalent to the existence
of a degree-$\ell$ curve through the columns indexed by $B$. If such a curve exists, then it must agree with the unique interpolation curve of degree at most $\ell$, through some fixed $\ell+1\le b$ columns of $M$ among those indexed by $B$. Since these columns lie in linear code $C$ and interpolation is linear, this curve can be taken to have coefficients
in $C$, i.e. a degree-$\ell$ codeword curve.
\end{proof}

We now cast curve-decoding as a row-span constrained LCL property, thereby bypassing the notion of $\bigvee$-decodability from \cite{gg25a}, and the parameter losses that arose from it.
\begin{theorem}[Locality of curve-decoding]
\label{thm:cd-casting}
Let $\ell,a,b\in\mathbb N$ with $\ell+1\le b\le a\le q$ and let
$\delta\in(0,1)$. Then there exists a row-span constrained
$a$-LCL family $\mc{F}$ with
\[
	|\mc{F}|\le q^{a+a^2}2^{an}
\]
such that every $\mb{F}_q$-additive code
$C$ is not
$(\ell,\delta,a,b)$ curve-decodable if and only if $C$ contains
$\mc{F}$.
\end{theorem}

\begin{proof}
Fix indexing sets $A \in \binom{\mb{F}_q}{a}$ and $I_\alpha \subset [n]$ with $|I_\alpha| \geq (1-\delta) n$ for each $\alpha \in A$.
For each choice of $A, \{I_\alpha\}_{\alpha \in A}$, we will define an $r$-local profile $\mbf{V}=\mbf{V}(A, \{I_\alpha\}_{\alpha\in A}) = (V_1, \ldots, V_n)$ as follows. For each $i\in[n]$, let
$A_i\coloneqq\{\alpha\in A:i\in I_\alpha\}$.
Define 
\begin{equation}\label{eq:Vi-def}
	V_i
	\coloneqq
	\left\{
	y\in\mb{F}_q^A:
	(y_\alpha)_{\alpha\in A_i}\in \on{RS}_{\ell}(A_i)
	\right\}.
\end{equation}
Equivalently, $y\in V_i$ if there exists a polynomial
$p\in\mb{F}_q[x]$ of degree at most $\ell$ such that
$p(\alpha)=y_\alpha$ for every $\alpha\in A_i$. This is clearly a linear subspace of
$\mb{F}_q^A$. 
Now define 
\begin{equation}
 \mc{F}\coloneqq \left\{\left(\mbf{V}(A, \{I_\alpha\}_{\alpha\in A}), U\right): A \in \binom{\mb{F}_q}{a}, \, I_\alpha \in \binom{[n]}{\geq (1-\delta) n }, \, U\in\mc{L}(\mb{F}_q^A) \text{ is $(\ell, A, b)$ curve-free}\right\}.
\end{equation}

We show the bound on the size of $\mc{F}$: There are $\binom{q}{a}\le q^a$ choices for $A$, at most $2^n$ choices for each $I_\alpha\subset [n]$, and the number of choices of $U$ is at most $|\mc{L}(\mb{F}_q^a)|\le q^{a^2}$.

Next, we prove that not being curve-decodable is equivalent to containing $\mathcal{F}$.

\paragraph{Not curve-decodable $\Rightarrow$ Contains $\mathcal{F}$:} 

If $C\subset(\mb{F}_q^s)^n$ is not $(\ell,\delta,a,b)$ curve-decodable, then there exists $A\in\binom{\mb{F}_q}{a}$, a degree-$\ell$ curve
$u:\mb{F}_q\to (\mb{F}_q^{s})^n$, and codewords
$f(\alpha)\in C$ such that
$\Delta(u(\alpha),f(\alpha))\le\delta$ for every $\alpha\in A$, but no
degree-$\ell$ codeword curve passes through $f(\alpha)$ for
$b$-many of $\alpha$. 
For $\alpha\in A$, define the agreement set
$I_\alpha\coloneqq\{i\in[n]:u(\alpha)_i=f(\alpha)_i\}$, so
$|I_\alpha|\ge(1-\delta)n$. Form the matrix
$M\in\mb{F}_q^{sn\times A}$ with columns $M_{\star\alpha}=f(\alpha)\in C$. Index the rows by $[n]\times [s]$ and let
$U$ be the row-span of $M$. 

We check the local profile constraints. For any $(i, j)\in[n]\times [s]$, the corresponding row $M_{ij, \star}=
(f(\alpha)_{ij})_{\alpha\in A}$. Since
$f(\alpha)_{ij}=u(\alpha)_{ij}$ whenever $i\in I_\alpha$, we observe that $M_{ij, \star}$ agrees with a polynomial $x\mapsto u(x)_{ij}$ of degree at most $\ell$ on $A_i$, so $M_{ij, \star}\in V_i$ by the definition of $V_i$ in \cref{eq:Vi-def}.

By assumption, no $b$ of the columns $f(\alpha)$ lie on a degree-$\ell$
codeword curve. By \cref{cor:curve-free-rowspan}, the row-span $U$ is
$(A,\ell,b)$ curve-free. Hence
$M$ witnesses that $C$ contains $(\mbf{V}(A, \{I_\alpha\}_{\alpha\in A}),U)\in \mc{F}$.

\paragraph{Contains $\mc{F}\Rightarrow$ Not curve-decodable:} 
Suppose $C$ 
contains the $a$-local pair $(\mbf{V}(A, \{I_\alpha\}_{\alpha\in A}),U)$ via witness matrix $M\in\mb{F}_q^{sn\times A}$. Let the columns of $M$ be $f(\alpha)\in C$ for $\alpha\in A$.

We construct a degree-$\ell$ curve close to these codewords. For $(i, j)\in[n]\times [s]$ the row $M_{ij,\star}$ lies in $V_i$, so by definition \cref{eq:Vi-def} of
$V_i$, there exists a polynomial $u_{ij}$ of degree
at most $\ell$ such that $u_{ij}(\alpha)=f(\alpha)_{ij}$ for every
$\alpha\in A_i\coloneqq \{\alpha\in A:i\in I_\alpha\}$. Collecting
these polynomials over all $(i,j)$ defines a degree-$\ell$ curve
$u:\mb{F}_q\to (\mb{F}_q^s)^n$.
For each $\alpha\in A$, if
$i\in I_\alpha$ so $\alpha\in A_i$, then
$u_{ij}(\alpha)=f(\alpha)_{ij}$ for every $j\in[s]$. Hence,
$u(\alpha)_i=f(\alpha)_i$ for all $i\in I_\alpha$. For every $\alpha\in A$, $|I_\alpha|\ge (1-\delta)n$, so $\Delta(u(\alpha),f(\alpha))\le\delta$.

Finally, since $U$ is $(A,\ell,b)$ curve-free, \cref{cor:curve-free-rowspan}
implies that no $b$ of the columns $f(\alpha)$ lie on a degree-$\ell$ codeword
curve. Together, $u$ and $f(\alpha)$ witness that $C$ is not
$(\ell,\delta,a,b)$ curve-decodable.
\end{proof}
We now plug in specific choices of parameters from \cref{thm:sdc-good-cd} for which subspace design codes are curve-decodable, thereby obtaining a lower bound on the threshold rate via \cref{cor:sdc-threshold-family}.
\begin{corollary}
\label{cor:cd-casting}
Let $\delta,\eta\in (0, 1)$ and integers $\ell, a, b\in\mb{N}$ such that
\begin{equation}
\label{eq:param-abd-cond}
	d\coloneqq \left\lceil \frac{\ell+1}{\eta}\right\rceil \le a \quad\text{and}\quad	b\le \frac{\eta a}{d+\eta}.
\end{equation}
Let $\mc{F}$ be the row-span constrained $a$-local LCL family
from \cref{thm:cd-casting}, so avoiding $\mc{F}$ is equivalent to being
$(\ell,\delta,a,b)$ curve-decodable. Then, $R_{\mc{F}}
	\ge
	1-\delta-\eta$. Additionally, $|\mathcal F|\leq q^{a+a^2}2^{an}$.
\end{corollary}
\begin{proof}
We claim that every $\tau$-subspace design code $C$ with $\tau(a)\le 1-\delta-\eta$ must avoid $\mc{F}$. Indeed, as $d\le a$ and $\tau$ is monotone, we know that 
\begin{equation}\label{eq:sdc-ab-choice}
\tau(d)\le\tau(a)\le 1-\delta-\eta.
\end{equation}
By \cref{thm:sdc-good-cd}, $C$ is $(\ell, 1-\tau(d)-\eta, a, \eta a/(d+\eta))$ curve-decodable. As $b\le \eta a/(d+\eta)$ and $\delta\le 1-\tau(d)-\eta$ by \cref{eq:sdc-ab-choice}, $C$ is $(\ell, \delta, a, b)$ curve-decodable, so it avoids $\mc{F}$.

Suppose $R_\mc{F}<1-\delta-\eta$, then there exists constant $\gamma>0$ such that $R_\mc{F}+\gamma \le1-\delta-\eta$. By \cref{cor:sdc-threshold-family}, there exists a $\tau$-subspace design code with $\tau(a)\le R_\mc{F}+\gamma \le1-\delta-\eta$ that contains $\mc{F}$. This contradicts the claim above, so it must be $R_\mc{F}\ge 1-\delta-\eta$, as desired.
\end{proof}

We conclude with a final specialization for the common choice of curve-decoding parameters.

\begin{corollary}
\label{cor:common-cd-params}
Let $\ell\in\mathbb N$ and $\delta, \eta\in(0,1)$. Define $a, b, d\in\mb{N}$ by
\begin{equation}
\label{eq:common-cd-params}
	d\coloneqq \left\lceil \frac{\ell+1}{\eta}\right\rceil,\qquad
	b\coloneqq \left\lceil \frac{4\ell}{\eta}+\ell\right\rceil,\qquad
	a\coloneqq \left\lceil \frac{100\ell^2}{\eta^3}\right\rceil.
\end{equation}
Let $\mc{F}$ be the row-span constrained $a$-local LCL family from \cref{thm:cd-casting}, so avoiding $\mc{F}$is equivalent to being
$(\ell,\delta,a,b)$ curve-decodable. Then, $R_{\mc{F}}
	\ge
	1-\delta-\eta$. Additionally, $|\mathcal F|\leq q^{a+a^2}2^{an}$.
\end{corollary}
\begin{proof}
It suffices to verify \cref{eq:param-abd-cond} to apply \cref{cor:cd-casting}. Clearly, $d\le a$. We check that
\begin{equation}
d+\eta\le \frac{\ell+1}{\eta}+1+\eta\le \frac{4\ell}{\eta}\implies \frac{\eta a}{d+\eta}\ge \frac{100\ell^2/\eta^2}{4\ell/\eta } =\frac{25\ell}{\eta}\ge \frac{4\ell}{\eta}+\ell+1\ge b.
\end{equation}
Therefore, \cref{cor:cd-casting} applies to give the conclusion.
\end{proof}

\section{Improved Proximity Gaps}
\label{sec:pg}

In this section we combine the result from \cref{sec:cd-local} that curve-decoding is a row-span constrained LCL property with list-decoding results for random linear codes, random Reed--Solomon codes, and Gallager's ensemble of random
LDPC codes, in order to establish proximity gaps (as well as correlated agreement and mutual correlated agreement) for these ensembles. For each ensemble of codes, we use
\cref{cor:common-cd-params} to prove curve-decodability, then we combine this with
list-decodability via \cref{cor: main-ca-mca-gg}.

\subsection{Random Linear Codes}
\label{sec:rlc}

We begin by citing the list-decoding fact we use for random linear codes. This is immediate from \cite{ld} up to plugging in some parameters.

\begin{theorem}[{\cite[Theorem 1.3]{ld}}]\label{thm:rlc-ld}
For every $\eta\in(0,1)$ and random linear code $C\subset\mb{F}_q^n$
 of rate $R$, if $q\ge 2^{\Omega(1/\eta^2)}$ and $n\ge \Omega(1/\eta)$, then with high probability $C$ is
\begin{equation}
	\left(1-R-\frac{3\eta}{2},\left\lceil\frac{1-R}{\eta}\right\rceil\right)
	\text{ list-decodable}.
\end{equation}
\end{theorem}
Next, we put together the list-decoding and curve-decoding results required to apply \cref{cor: main-ca-mca-gg} for random linear codes. 
\begin{theorem}\label{thm:cd+ld-rlc}
Let $\ell\in\mb{N}$ and $\eta,\delta\in(0,1)$. Let $C\subset \FF_q^n$ be a
random linear code of rate $R\le 1-\delta-2\eta$.  If $q\ge \exp(\Omega(\ell^2/\eta^4))$ and $n\ge\Omega(\ell^4/\eta^7)$, then with high probability $C$ is both
\begin{equation}
\left(\delta\left(1+\frac{\eta}{2}\right),\left\lceil\frac{1-R}{\eta}\right\rceil\right)
	\text{ list-decodable}
\quad\text{and}\quad
	\left(
	\ell,\delta,
	\left\lceil\frac{100\ell^2}{\eta^3}\right\rceil,
	\left\lceil\frac{4\ell}{\eta}+\ell\right\rceil
	\right)
	\text{ curve-decodable}.
\end{equation}
\end{theorem}

\begin{proof}
Let $a,b,d$ be as in \cref{cor:common-cd-params}, and let $\mc{F}$ be the $a$-LCL property whose avoidance encodes $(\ell,\delta,a,b)$ curve-decoding constructed in \cref{thm:cd-casting}.  By
\cref{cor:common-cd-params}, $R_{\mc{F}}\ge 1-\delta-\eta$, so
\begin{equation}
	R\le 1-\delta-2\eta\le R_{\mc{F}}-\eta.
\end{equation}
Therefore, by \cref{cor:rlc-threshold-family},
\begin{equation}
	\mb{P}(C\text{ contains }\mc{F})
	\le
	|\mc{F}|q^{-\eta n}
	\le
	q^{a^2+a-\eta n}2^{an}.
\end{equation}
Since $a=O(\ell^2/\eta^3)$, the right-hand side is $o(1)$ whenever $q$ is sufficiently large, as assumed in the theorem. Hence, with high
probability, $C$ avoids $\mc{F}$, so it is
$(\ell,\delta,a,b)$ curve-decodable.

For list-decoding, observe that as $R\le 1-\delta-2\eta$,
\begin{equation}
	\delta\left(1+\frac{\eta}{2}\right)
	\le
	\delta+\frac{\eta}{2}
	\le
	1-R-\frac{3\eta}{2}.
\end{equation}
For $q$ sufficiently large per the assumption in the theorem,
\cref{thm:rlc-ld} implies that, with high probability, $C$ is
$\left(\delta(1+\eta/2),\lceil(1-R)/\eta\rceil\right)$ list-decodable.

Taking a union bound over the list-decoding and curve-decoding events
completes the proof.
    \end{proof}

We are now ready to prove our main theorem on proximity gaps, correlated agreement, mutual correlated agreement of random linear codes.

\begin{theorem}[PG, CA, and MCA for Random Linear Codes]\label{thm:rlc-main}
Let $\ell\in\mb{N}$ and $\eta,\delta\in(0,1)$. Let $C\subset \FF_q^n$ be a
random linear code of rate $R\le 1-\delta-2\eta$.  If $q\ge \exp(\Omega(\ell^2/\eta^4))$ and $n\ge\Omega(\ell^4/\eta^7)$, then with high probability:
\begin{enumerate}
    \item For every $T\ge 1$, $C$ is curve-decodable with parameters
    \begin{equation}
    	\left(
    	\ell,\delta,
    	(T-1)\left\lceil\frac{1-R}{\eta}\right\rceil
    	+O\left(\frac{\ell^2}{\eta^3}\right),
    	T
    	\right).
    \end{equation}
    \item For every $m > 1$, $C$ has correlated agreement, and hence proximity gap, with parameters
    \begin{equation}
    	\left(
    	\ell,\delta,
    	\frac{
    	\ell m\left\lceil(1-R)/\eta\right\rceil
    	+O(\ell^2/\eta^3)}
    	{q},
    	\frac{1}{m}
    	\right).
    \end{equation}
    \item $C$ has mutual correlated agreement with parameters
    \begin{equation}
    	\left(
    	\ell,\delta,
    	\frac{
    	\ell n\left\lceil(1-R)/\eta\right\rceil
    	+O(\ell^2/\eta^3)}
    	{q}
    	\right).
    \end{equation}
\end{enumerate}
\end{theorem}

\begin{proof}
Let $a,b$ be as in \cref{cor:common-cd-params}, and let
$L=\lceil(1-R)/\eta\rceil$. By \cref{thm:cd+ld-rlc}, 
$C$ is $(\delta(1+\eta/2),L)$ list-decodable and
$(\ell,\delta,a,b)$ curve-decodable with high probability.
By \cref{cor:common-cd-params}, we observe that $\delta\left(1+{\ell}/(b-\ell)\right)
	\le
	\delta\left(1+{\eta}/{2}\right)$, and therefore the hypotheses of \cref{cor: main-ca-mca-gg} hold with this value of
$L$.  Since $a=O(\ell^2/\eta^3)$, the three conclusions follow directly from
\cref{cor: main-ca-mca-gg}.
\end{proof}

\subsection{Random Reed-Solomon Codes}
\label{sec:rrs}
Our results naturally generalize to Reed-Solomon codes with random evaluation points (\cref{def: RRS}). The proof is
parallel to the random-linear-code case, except that the low-rate local
threshold uses \cref{cor:rrs-threshold-family} and the list-decoding input is
the following result from \cite{ld}.

\begin{theorem}[{\cite[Theorem 1.1]{ld}}]\label{thm:rrs-ld}
For every $\eta\in(0,1)$ and random Reed-Solomon code $C\subset\mb{F}_q^n$ of rate $R$, if $q\ge n\cdot 2^{\Omega(1/\eta^2)}$, then with high probability $C$ is
\begin{equation}
	\left(1-R-\frac{3\eta}{2},\left\lceil\frac{1-R}{\eta}\right\rceil\right)
	\text{ list-decodable}.
\end{equation}
\end{theorem}
Next, we put together the list-decoding and curve-decoding results required to apply \cref{cor: main-ca-mca-gg} for random Reed-Solomon codes. 

\begin{theorem}\label{thm:cd+ld-rrs}
Let $\ell\in\mb{N}$ and $\eta,\delta\in(0,1)$. Let $C\subset \FF_q^n$ be a
random Reed-Solomon code of rate $R\le 1-\delta-2\eta$.  If $q\ge n\cdot  \exp(\Omega(\ell^4/\eta^7))$ and $n\ge\Omega(\ell^4/\eta^7)$, then with high probability $C$ is both
\begin{equation}
\left(\delta\left(1+\frac{\eta}{2}\right),\left\lceil\frac{1-R}{\eta}\right\rceil\right)
	\text{ list-decodable}
\quad\text{and}\quad
	\left(
	\ell,\delta,
	\left\lceil\frac{100\ell^2}{\eta^3}\right\rceil,
	\left\lceil\frac{4\ell}{\eta}+\ell\right\rceil
	\right)
	\text{ curve-decodable}.
\end{equation}
\end{theorem}

\begin{proof}
Let $a,b,d$ be as in \cref{cor:common-cd-params}, and let $\mc{F}$ be the $a$-LCL property whose avoidance encodes $(\ell,\delta,a,b)$ curve-decoding constructed in \cref{thm:cd-casting}.  By
\cref{cor:common-cd-params}, $R_{\mc{F}}\ge 1-\delta-\eta$, so
\begin{equation}
	R\le 1-\delta-2\eta\le R_{\mc{F}}-\eta.
\end{equation}
Therefore, by \cref{cor:rrs-threshold-family} and that $|\mc{F}|\le q^{a^2+a}2^{an}$,
\begin{equation}
\mb{P}(C\text{ contains }\mc{F})
\le
|\mc{F}|(2^a-1)
\left(\frac{(4a)^{4a}Rn}{\eta q}\right)^{\eta n/2a}.
\end{equation}
Since $a=O(\ell^2/\eta^3)$, the right-hand side is $o(1)$ whenever $q$ is sufficiently large per the assumption in the theorem. Hence, with high
probability, $C$ avoids $\mc{F}$, so it is
$(\ell,\delta,a,b)$ curve-decodable.

For list-decoding, observe that as $R\le 1-\delta-2\eta$,
\begin{equation}
	\delta\left(1+\frac{\eta}{2}\right)
	\le
	\delta+\frac{\eta}{2}
	\le
	1-R-\frac{3\eta}{2}.
\end{equation}
For $q$ sufficiently large as assumed in the theorem,
\cref{thm:rrs-ld} implies that, with high probability, $C$ is
$\left(\delta(1+\eta/2),\lceil(1-R)/\eta\rceil\right)$ list-decodable.

Taking a union bound over the list-decoding and curve-decoding events
completes the proof.
    \end{proof}

We are now ready to prove our main theorem on proximity gaps, correlated agreement, mutual correlated agreement of random Reed-Solomon codes.

\begin{theorem}[PG, CA, and MCA for Random Reed-Solomon Codes]\label{thm:rrs-main}
Let $\ell\in\mb{N}$ and $\eta,\delta\in(0,1)$. Let $C\subset \FF_q^n$ be a
random Reed-Solomon code of rate $R\le 1-\delta-2\eta$.  If $q\ge n\cdot \exp(\Omega(\ell^4/\eta^7))$ and $n\ge\Omega(\ell^4/\eta^7)$, then with high probability:
\begin{enumerate}
    \item For every $T\ge 1$, $C$ is curve-decodable with parameters
    \begin{equation}
    	\left(
    	\ell,\delta,
    	(T-1)\left\lceil\frac{1-R}{\eta}\right\rceil
    	+O\left(\frac{\ell^2}{\eta^3}\right),
    	T
    	\right).
    \end{equation}
    \item For every $m > 1$, $C$ has correlated agreement, and hence proximity gap, with parameters
    \begin{equation}
    	\left(
    	\ell,\delta,
    	\frac{
    	\ell m\left\lceil(1-R)/\eta\right\rceil
    	+O(\ell^2/\eta^3)}
    	{q},
    	\frac{1}{m}
    	\right).
    \end{equation}
    \item $C$ has mutual correlated agreement with parameters
    \begin{equation}
    	\left(
    	\ell,\delta,
    	\frac{
    	\ell n\left\lceil(1-R)/\eta\right\rceil
    	+O(\ell^2/\eta^3)}
    	{q}
    	\right).
    \end{equation}
\end{enumerate}
\end{theorem}

\begin{proof}
The proof is identical to that of \cref{thm:rlc-main}, using
\cref{thm:cd+ld-rrs} in place of \cref{thm:cd+ld-rlc}.
\end{proof}

\subsection{Random LDPC Codes}
\label{sec:rlpdc}
Our results also naturally generalize to random LDPC codes from Gallager's ensemble (\cref{def: RLDPC}). The proof is
parallel to the random-linear-code case: the low-rate local
threshold is supplied by
\cref{cor:rldpc-threshold-family}, and the list-decoding input follows from
\cite[Theorem 1.2]{mrrsw20}, which transfers list-decodability of random
linear codes (\cref{thm:rlc-ld}) to sufficiently sparse random LDPC codes.
We begin by recording the list-decoding fact we use.  
\begin{theorem}
\label{thm:rldpc-ld}
For every $\eta, R\in(0,1)$, there exists
$s_0=s_0(\eta,q,R)$ such that the following holds for every odd
$s>s_0$.  Let $C=\mathsf{RLDPC}(n,q,s,R)$ be a random LDPC code from
Gallager's ensemble with rate $R$. If
$q\ge 2^{\Omega(1/\eta^2)}$ and $n$ is sufficiently large, then with high
probability $C$ is
\begin{equation}
    \left(
        1-R-\frac{3\eta}{2},
        \left\lceil\frac{2(1-R)}{\eta}\right\rceil
    \right)
    \text{ list-decodable}.
\end{equation}
\end{theorem}
This is
an immediate consequence of \cite[Theorem 1.2]{mrrsw20} and
\cref{thm:rlc-ld}, which is exactly {\cite[Theorem 1.3]{ld}}.  

\begin{remark}\label{rem:factor of two}The factor of two loss in the list size of \cref{thm:rldpc-ld} compared to \cref{thm:rlc-ld} comes from the
small slack in the rate in \cite[Theorem 1.2]{mrrsw20}, and can be made arbitrarily small.
\end{remark}

Next, we put together the list-decoding and curve-decoding results required to apply \cref{cor: main-ca-mca-gg} for random LDPC codes. 

\begin{theorem}\label{thm:cd+ld-rldpc}
Let $\ell\in\mb{N}$ and $\eta,\delta\in(0,1)$, and let $R \leq 1 -\delta - 2\eta$. There exists
$s_0=s_0(\ell,\eta,\delta,q,R)$ such that the following holds for every odd
$s>s_0$. Let $C=\mathsf{RLDPC}(n,q,s,R)$ be a random LDPC code from
Gallager's ensemble with rate $R$. If
$q\ge \exp(\Omega(\ell^2/\eta^4))$ and
$n$ is sufficiently large, then with high probability $C$ is both
\begin{equation}
\left(\delta\left(1+\frac{\eta}{2}\right),\left\lceil\frac{2(1-R)}{\eta}\right\rceil\right)
	\text{ list-decodable}
\quad\text{and}\quad
	\left(
	\ell,\delta,
	\left\lceil\frac{100\ell^2}{\eta^3}\right\rceil,
	\left\lceil\frac{4\ell}{\eta}+\ell\right\rceil
	\right)
	\text{ curve-decodable}.
\end{equation}
\end{theorem}

\begin{proof}
Let $a,b,d$ be as in \cref{cor:common-cd-params}, and let $\mc{F}$ be the $a$-LCL property whose avoidance encodes $(\ell,\delta,a,b)$ curve-decoding constructed in \cref{thm:cd-casting}.  By
\cref{cor:common-cd-params}, $R_{\mc{F}}\ge 1-\delta-\eta$, so
\begin{equation}
	R\le 1-\delta-2\eta\le R_{\mc{F}}-\eta.
\end{equation}
By \cref{cor:rldpc-threshold-family}, for sufficiently large odd
$s$, we get
\begin{equation}
	\mb{P}(C\text{ contains }\mc{F})
	\le
	|\mc{F}|q^{-\eta n/2}+o_{n\to\infty}(1)
	\le
	q^{a^2+a-\eta n/2}2^{an}+o_{n\to\infty}(1).
\end{equation}
Since $a=O(\ell^2/\eta^3)$, the right-hand side is $o(1)$ whenever $q$ is sufficiently large, as assumed in the theorem. Hence, with high
probability, $C$ avoids $\mc{F}$, so it is
$(\ell,\delta,a,b)$ curve-decodable.

For list-decoding, observe that as $R\le 1-\delta-2\eta$,
\begin{equation}
	\delta\left(1+\frac{\eta}{2}\right)
	\le
	\delta+\frac{\eta}{2}
	\le
	1-R-\frac{3\eta}{2}.
\end{equation}
For $q$ sufficiently large per the assumption in the theorem,
\cref{thm:rldpc-ld} implies that, with high probability, $C$ is
$\left(\delta(1+\eta/2),\lceil 2(1-R)/\eta\rceil\right)$ list-decodable.

Taking a union bound over the list-decoding and curve-decoding events
completes the proof.
    \end{proof}

We are now ready to prove our main theorem on proximity gaps, correlated agreement, mutual correlated agreement of random LDPC codes.

\begin{theorem}[PG, CA, and MCA for Random LDPC Codes]\label{thm:rldpc-main}
Let $\ell\in\mb{N}$ and $\eta,\delta\in(0,1)$. There exists
$s_0=s_0(\ell,\eta,\delta,q)$ such that the following holds for every odd
$s>s_0$. Let $C=\mathsf{RLDPC}(n,q,s,R)$ be a random LDPC code from
Gallager's ensemble with rate $R\le 1-\delta-2\eta$. If
$q\ge \exp(\Omega(\ell^2/\eta^4))$ and
$n$ is sufficiently large, then with high probability:
\begin{enumerate}
    \item For every $T\ge 1$, $C$ is curve-decodable with parameters
    \begin{equation}
        \left(
        \ell,\delta,
        (T-1)\left\lceil\frac{2(1-R)}{\eta}\right\rceil
        +O\left(\frac{\ell^2}{\eta^3}\right),
        T
        \right).
    \end{equation}

    \item For every $m > 1$, $C$ has correlated agreement, and hence
    proximity gap, with parameters
    \begin{equation}
        \left(
        \ell,\delta,
        \frac{
        \ell m\left\lceil 2(1-R)/\eta\right\rceil
        +O(\ell^2/\eta^3)}
        {q},
        \frac{1}{m}
        \right).
    \end{equation}

    \item $C$ has mutual correlated agreement with parameters
    \begin{equation}
        \left(
        \ell,\delta,
        \frac{
        \ell n\left\lceil 2(1-R)/\eta\right\rceil
        +O(\ell^2/\eta^3)}
        {q}
        \right).
    \end{equation}
\end{enumerate}
\end{theorem}

\begin{proof}
The proof is identical to that of \cref{thm:rlc-main}, using
\cref{thm:cd+ld-rldpc} in place of \cref{thm:cd+ld-rlc}.
\end{proof}

\bibliographystyle{alpha}
\bibliography{ref.bib}

\appendix

\section{Appendix: Gallager's Ensemble}\label{app: RLDPC}

In this appendix we prove \cref{lem: RLDPC-threshold}, the fixed-pair
containment estimate for row-span constrained local profiles in Gallager's
ensemble.  The proof follows the same philosophy as \cite{mrrsw20}.  
\begin{enumerate}
    \item 
\cite{mrrsw20} proves that a fixed smooth matrix is contained in a random LDPC
code with essentially the same probability as in a random linear code. 
\item Random LDPC codes have good distance with high probability ({\cite[Theorem 2.14]{mrrsw20}}, stated as \cref{lem:rldpc-good-distance} earlier), so non-smooth matrices will not be contained after conditioning on the event
\[
    \mathsf{Good}
    \coloneqq
    \left\{
        \Delta(C)\ge \frac{1}{2}h_q^{-1}(1-R)
    \right\}.
\]
\end{enumerate}
Combining these two facts gives the matrix-containment estimate that we aim to prove. 
We begin with the notion of smoothness used in \cite{mrrsw20}.

\begin{definition}[$\delta$-smooth matrices]
$M\in\FF_q^{n\times r}$ is \emph{$\delta$-smooth} if for
every nonzero $x\in\FF_q^r$,
\begin{equation}
    \Delta(Mx,0)\ge \delta.
\end{equation}
Equivalently, the linear code generated by the columns of $M$ has relative
distance at least $\delta$. In particular, a $\delta$-smooth matrix has full
column rank.
\end{definition}

The next input is the fixed smooth matrix containment estimate from
\cite{mrrsw20}. Here $M\subset C$ means that every column of $M$ is contained in $C$.

\begin{lemma}[{\cite[Lemma 2.13]{mrrsw20}}]
\label{lem:rldpc-smooth-matrix}
For every $\delta,\xi>0$ and $r\in\mb{N}$, there
exists $s_0=s_0(\delta,\xi,q,r)$ such that the following holds for every odd
$s>s_0$ and all sufficiently large $n$.
Let $M\in\FF_q^{n\times r}$ be a $\delta$-smooth matrix. If
$C=\mathsf{RLDPC}(n,q,s,R)$, then
\begin{equation}
    \mb{P}\left(M\subset C\right)
    \le
    q^{(1-\xi)(R-1)rn}.
\end{equation}
\end{lemma}

We now upgrade the smooth-matrix estimate to an arbitrary fixed matrix,
conditioned on $\mathsf{Good}$.

\begin{lemma}[Fixed matrix containment under good distance]
\label{lem:rldpc-fixed-matrix}
For every $\eps,R\in(0,1)$, prime power $q$, and positive integer $r$, there
exists $s_0=s_0(\eps,r,R,q)$ such that for every odd $s>s_0$, the following
holds for all sufficiently large $n$.
If $C=\mathsf{RLDPC}(n,q,s,R)$, then for every matrix $M\in\FF_q^{n\times r}$,
\begin{equation}
    \mb{P}\left(
        M\subset C
       \middle|
        \mathsf{Good}
    \right)
    \le
    q^{(R-1)n\on{rank}(M)+\eps n}\quad\text{where}\quad   \mathsf{Good}
    \coloneqq
    \left\{
        \Delta(C)\ge \frac{1}{2}h_q^{-1}(1-R)
    \right\}.
\end{equation}
\end{lemma}

\begin{proof}
Let $r'=\operatorname{rank}(M)$. Choose $r'$ linearly independent columns of
$M$ and let $N\in\FF_q^{n\times r'}$ be the corresponding submatrix. Since
$C$ is linear, $M\subset C$ if and only if $N\subset C$. Fix $\delta\coloneqq h_q^{-1}(1-R)/2$.

If $N$ is not $\delta$-smooth, then there exists a nonzero
$x\in\FF_q^{r'}$ such that $\Delta(Nx,0)<\delta$. If $N\subset C$, then
$Nx$ is a nonzero codeword of $C$ of relative weight less than $\delta$,
contradicting $\mathsf{Good}$. Thus in this case $\mb{P}\left(N\subset C\middle|\mathsf{Good}\right)=0$, which trivially satisfies the lemma as $M\subset C$ if and only if $N\subset C$.

If $N$ is $\delta$-smooth, then by \cref{lem:rldpc-smooth-matrix} with
$\xi=\eps/(2r)$ and that $r'\le r$, we bound
\begin{equation}
\mb{P}(N\subset C) \le
    q^{-(1-\xi)(1-R)r'n}
    \le
    q^{(R-1)r'n+\eps n/2}.
\end{equation}
By \cref{lem:rldpc-good-distance}, after increasing $s_0$ if necessary,
$\mb{P}(\mathsf{Good})\ge 1-o_{n\to\infty}(1)$. Hence, 
\begin{equation}
    \mb{P}\left(
        N\subset C
       \middle| 
        \mathsf{Good}
    \right) \le 
    \frac{\mb{P}(N\subset C)}{\mb{P}(\mathsf{Good})}
    \le
    q^{(R-1)r'n+\eps n}
\end{equation}
for sufficiently large $n$.
Since $M\subset C$ if and only if $N\subset C$, the lemma follows.
\end{proof}

We now prove \cref{lem: RLDPC-threshold}, restated below for the reader's convenience.

\rldpclocalprofiles*

\begin{proof}
Let $\mc{M}_{\mbf{V}}(U)$
be the set of matrices $M\in\FF_q^{n\times r}$ with row-span exactly $U$
such that $M_{i\star}\in V_i$ for every $i\in[n]$. If $C$ contains
$(\mbf{V},U)$, then it contains some matrix $M\in\mc{M}_{\mbf{V}}(U)$.
Every matrix $M\in\mc{M}_{\mbf{V}}(U)$ has rank $r'=\dim U$. Therefore, by
\cref{lem:rldpc-fixed-matrix},
\begin{equation}\label{eq:A1}
    \mb{P}\left(
        M\subset C
       \middle| 
        \mathsf{Good}
    \right)
    \le
    q^{(R-1)nr'+\eps n}.
\end{equation}
By \cref{lem:profile-size}, we have
\begin{equation}\label{eq:A2}
    |\mc{M}_{\mbf{V}}(U)|
    \le
    q^{n((1-R)\dim U+\phi_{\mbf{V}}(U,R))}.
\end{equation}
Union bounding over all $M\in\mc{M}_\mbf{V}(U)$ and combining \cref{eq:A1,eq:A2} gives
\begin{equation}
\mb{P}\left(C\text{ contains }(\mbf{V},U)\middle| \mathsf{Good}
\right)
\le
q^{n((1-R)r'+\phi_{\mbf{V}}(U,R))}
q^{(R-1)nr'+\eps n} =
q^{n\phi_{\mbf{V}}(U,R)+\eps n}.
\end{equation}
This proves the lemma.
\end{proof}

\end{document}